\documentclass[10pt]{article}
\usepackage{amsmath, amssymb,mathrsfs}
\usepackage{color, graphics}
\usepackage{amsthm}
\newtheorem{lem}{Lemma}
\newtheorem{prop}{Proposition}
\newtheorem{thm}{Theorem}
\newtheorem{cor}{Corollary}

\newtheorem{ass}{Assumption}

\newtheorem{defi}{Definition}
\newtheorem{nota}{Notation}

\newtheorem{rem}{Remark}
\usepackage{url}

\begin{document}
%
\title{Notes on Backward Stochastic Differential Equations for Computing XVA}
%
%
\author{Jun Sekine\footnote{
	Graduate School of Engineering Science,
	Osaka University,
	1-3, Machikaneyama-cho, Toyonaka,
	Osaka, 560-8531, Japan 
	Email: \texttt{sekine@sigmath.es.osaka-u.ac.jp}
}
	\footnote{
Jun Sekine's research is supported by 
a Grant-in-Aid for Scientific Research (C), No. 19K03636, 
from the Japan Society for the Promotion of Science.} 
\quad and \quad
Akihiro Tanaka \footnote{
	Graduate School of Engineering Science,
	Osaka University,
	1-3, Machikaneyama-cho, Toyonaka,
	Osaka, 560-8531, Japan / 
	Sumitomo Mitsui Banking Corporation,
	1-1-2, Marunouchi, Chiyoda-ku, 
	Tokyo, 100-0005, Japan,
	Email: \texttt{tnkaki2000@gmail.com}
}
}
\date{}

\maketitle              

\begin{abstract}
The X-valuation adjustment (XVA) problem, 
which is a recent topic in mathematical finance, is 
considered and analyzed. 
First, the basic properties of backward stochastic differential equations (BSDEs)
with a random horizon in a progressively enlarged filtration 
are reviewed. 
Next, the pricing/hedging problem for defaultable over-the-counter (OTC)
derivative securities is described using such BSDEs.
An explicit sufficient condition is given
to ensure the non-existence of an arbitrage opportunity
for both the seller and buyer of the derivative securities.
Furthermore, an explicit pricing formula is presented in which 
XVA is interpreted as approximated correction terms 
of the theoretical fair price.
\end{abstract}
\section{Introduction}

Backward stochastic differential equations 
(BSDEs) have been studied intensively 
from both theoretical and application viewpoints.
Bismut (1976, 1978) studied BSDEs related to stochastic control problems, 
and Pardoux and Peng (1990) introduced general {\it nonlinear} 
BSDEs driven by Brownian motion as a noise process.
After those early pioneering studies and since the late 1990s, 
the field of mathematical finance has provided various interesting 
research topics to develop the theory and application of BSDEs
(e.g., El Karoui et al., 2000).
In the present paper, 
we are interested in one such recent research topic 
in mathematical finance, 
namely, the X-valuation adjustment (XVA) problem.
The pricing and hedging methodology for over-the-counter (OTC) 
financial derivative securities 
for practitioners in financial institutions
has been modified since the global financial crisis in 2008.
The pre-crisis pricing was based on the Black--Scholes--Merton 
paradigm, and 
\[
 {\rm p}_{\rm RN}:= {\mathbb E}\left[ {\rm DF}_r(T) \xi_T\right]
\]
was regarded as the ``fair'' price 
of the derivative security $(T,\xi_T)$.
Here, $\xi_T$ is a random variable representing the 
payoff at the maturity date $T\in {\mathbb R}_{++}(:=(0,\infty))$
of the derivative security, 
${\rm DF}_r(T):=
\exp\left\{ -\int_0^T r(u)du\right\}$ is a suitable discounting factor,
where $r:=(r(t))_{t\ge 0}$ is a risk-free interest rate process, 
and ${\mathbb E}\left[ (\cdot)\right]$ 
represents the expectation with respect to the so-called 
risk-neutral probability measure. 
By contrast, the post-crisis pricing formula 
used by practitioners in financial institutions
is now described as
\begin{equation}
 \bar{\rm p}_{\rm RN} + \sum_{x} x{\rm VA} 
\end{equation}
for the derivative security $(T,\xi_T)$. Here, 
\[
 \bar{\rm p}_{\rm RN}:= {\mathbb E}\left[ 
{\rm DF}_{\bar{r}}(T) \xi_T\right], 
\]
employing $\bar{r}:=(\bar{r}(t))_{t\ge 0}$ 
as a risk-free interest rate process,
which is different from $r$ used in the pre-crisis
model,\footnote{
The London Interbank Offered Rate (LIBOR) was a popular choice 
as the risk-free rate in pre-crisis models, 
whereas the Overnight Index Swap (OIS) rate is now recognized as 
a suitable candidate as the risk-free rate 
in post-crisis models.} 
and 
\[
\sum_{x} x{\rm VA}=
{\rm CVA} - {\rm DVA} + {\rm FVA} + {\rm ColVA} + \cdots
\]
represents various valuation adjustments
(e.g., credit valuation adjustment, debt valuation adjustment, 
funding valuation adjustment, collateral valuation adjustment).
We may interpret the post-crisis modification as reflecting 
the following current situations.
\begin{itemize}
 \item[(a)] The credit risk (default risk) of investors and their
	    counterparties and the liquidity risk (of assets and cash)
are widely recognized and and now considered seriously.
 \item[(b)] As a consequence of (a), the differences in various interest rates 
(e.g., risk-free rate, repo rate, funding rate, collateral rate) 
can no longer be neglected.
\end{itemize}
In this paper, we aim to understand
the post-crisis pricing formula (1) in a better way
from a theoretical viewpoint. 
Using BSDEs, which model the value processes of hedging portfolios, 
we interpret (1) 
as an approximate value of the fair price (i.e., the replication cost) of 
a derivative security.
Concretely, this paper is organized as follows.
\begin{itemize}
 \item
In Section~2, we prepare a BSDE with a random horizon, 
where two random times $\tau_1,\tau_2$ 
and the progressively enlarged filtration by these random times
are introduced, and 
the horizon is set as $\tau_1\wedge \tau_2 \wedge T$ ($T\in {\mathbb R}_{++}$).
We review some basic properties of such a BSDE, 
that is, the existence of a unique solution  and its construction,  
using a reduced BSDE defined on a smaller filtration
(see Theorems~1--3).
These results are then used in Section~3.
 \item
In Section~3, we construct a financial market model 
that generalizes the model given by Bichuch et al.\ (2018). 
On it, we derive BSDEs for pricing and hedging derivative securities,
which express {\it nonlinear} dynamic hedging portfolio values 
of the seller and buyer.
Here, we model the default time of the hedger 
(i.e., the seller of a derivative security) $\tau_1$
and that of her counterparty 
(i.e., the buyer of the derivative security) $\tau_2$, 
each of which are defined by random times. 
The contract between the hedger and her counterparty expires
if the hedger or the counterparty defaults. Hence, 
$\tau_1 \wedge \tau_2 \wedge T$ is 
interpreted as the (random) horizon of the contract, 
where $T$ is the prescribed fixed maturity, 
and we naturally have BSDEs considered in Section~2.

 \item 
In Section~4, working with the BSDEs introduced in Section~3, 
we obtain the following. 
\begin{itemize}
 \item[(i)] 
An explicit sufficient condition
is presented to ensure the non-existence
of an arbitrage opportunity for both the seller and buyer of the derivative security
(see Theorem~4).
We note that a rather restrictive condition is necessary
to ensure the existence of an arbitrage-free price
(see Remark~14). 
 \item[(ii)]
The pricing formula (1) used by practitioners
is interpreted as an approximation of the theoretical fair price 
of the derivative security: 
XVA is regarded as certain ``zero-th'' order approximated correction terms.
(see Theorem~5, Corollary~1, Proposition~3, and Remark~16).
Furthermore, we mention a higher first-order approximation 
(see Subsection~4.3).
\end{itemize}
\end{itemize}
We intend to write this paper in an expository manner generally:
Section~2 is devoted for reviewing known results
and some results in Section~4 
(that is, Theorem~4 and Proposition~1 and 2)
are rather straightforward extensions of existing results
of the closely related work by 
Bichuch et al.\ (2015, 2018) and Tanaka (2019).
For other parts, we regard the following as being the contributions
of the paper in comparison with 
Bichuch et al.\ (2015, 2018) and Tanaka (2019).
\begin{itemize}
 \item[1)] The market model is generalized: our model treats 
\begin{itemize}
 \item[(i)] a multiple risky asset model, and 
 \item[(ii)] a stochastic factor model that includes 
a stochastic volatility, a stochastic interest rate, 
and a stochastic hazard rate.
\end{itemize}
 \item[2)] Different definitions of 
arbitrages and admissible trading strategies are employed
(see Subsection~3.5).
Because we analyze the pricing/hedging problem of derivative securities
by using BSDEs, our choices seem to be natural and clear. 
 \item[3)] For XVA, an interpretation of pricing formula (1)
is given as well as its arbitrage-free property
(see Theorem~5, Corollary~1, and Proposition~3 with 
the following Remark~16 in Subsection~4.2, and cf.\ the results in \cite{T}).
 \item[4)] 
Regarding the lending-borrowing spreads of interest rates as 
``small parameters'', 
the first order perturbed BSDEs are derived 
and the associated approximated valuation adjustment terms 
are computed (see Proposition~4 in Subsection~4.3).
\end{itemize}

\section{BSDE with a Random Horizon in a Progressively Enlarged Filtration}

\subsection{Setup}
Let $(\Omega, {\mathcal F},{\mathbb P})$
be a complete probability space and 
let $W:=\left(W(t)\right)_{t\ge 0}$, 
$W(t):=\left(W_1(t),\dots,W_n(t)\right)^\top$ 
be an $n$-dimensional Brownian motion on it.
Define the filtration by
\[
 {\mathcal F}_t :=\sigma\left(W(s); s\in [0,t]\right) \vee {\mathcal N},
\quad t\ge 0,
\]
where ${\mathcal N}$ is the totality of null sets. 
Let $E_1,E_2$ be exponentially distributed random variables, 
assuming that $W$, $E_1$, and $E_2$ are mutually independent. 
Using nonnegative ${\mathcal F}_t$-progressively measurable processes
$h_i:=\left(h_i(t)\right)_{t\ge 0}$, ($i=1,2$), 
define the random times $\tau_1, \tau_2$ by
\begin{equation}
 \tau_i:=\inf\left\{ t\ge 0 \Bigm| \int_0^t h_i(u) du \ge E_i \right\}.
\end{equation}
The indicator processes for $\tau_i$ ($i=1,2$), namely 
\[
 N_i(t):=1_{\{ t\ge \tau_i\}},
\quad t\ge 0, 
\]
are submartingales with respect to the filtration
\[
 {\mathcal H}_t:=\sigma\left( N_1(s), N_2(s) ; \ s\in [0,t]\right),
\quad t\ge 0, 
\]
and their Doob--Meyer decompositions are written as
\[
 N_i(t) =M_i(t) + 
\int_0^t \left\{ 1-N_i(s)\right\} h_i(s) ds,
\quad t\ge 0
\]
for $i=1,2$, 
where 
\[
M_i(t):= N_i(t) -\int_0^t \left\{ 1-N_i(s)\right\} h_i(s) ds,
\quad t\ge 0
\]
($i=1,2$)
are two independent martingales with respect to $({\mathcal H}_t)_{t\ge 0}$.
Moreover, $(W,M_1,M_2)$ remain as martingales
with respect to the progressively enlarged filtration, 
\[
 {\mathcal G}_t:={\mathcal F}_t \vee {\mathcal H}_t,
\quad t\ge 0
\]
(e.g., see Section~2.3 of ~Aksamit and Jeanblanc, 2017), 
which are mutually independent.  
Also, we deduce that for $0\le s\le t$, 
\[
 {\mathbb P}\left( \tau_i> s \bigm| {\mathcal F}_t\right)
={\mathbb P}\left( \tau_i> s \bigm| {\mathcal F}_\infty\right)
=\exp\left\{ -\int_0^s h_i(u)du\right\},
\]
where $\displaystyle
{\mathcal F}_\infty:=\sigma\left( \cup_{t\ge 0}{\mathcal F}_t\right)$.
From this, we see that for $ds\ll 1$, 
\begin{align*}
{\mathbb P}\left( \tau_i\le s+ds \bigm|\tau_i>s, {\mathcal F}_\infty \right)
=&
\frac{{\mathbb P}\left( s<\tau_i\le s+ds
| {\mathcal F}_\infty\right)}
{{\mathbb P}\left( \tau_i>s 
| {\mathcal F}_\infty\right)} \\
=&1-\exp\left\{ -\int_s^{s+ds} h_i(u)du\right\}
\approx
h_i(s)ds,
\end{align*}
and $h_i$ is called the hazard rate (or intensity) process for $\tau_i$.
Following Pham (2010), we employ the notation below.
\begin{nota}
\begin{itemize}
 \item ${\mathbb F}:=({\mathcal F}_t)_{t\ge 0}$, 
${\mathbb G}:=({\mathcal G}_t)_{t\ge 0}$, and
${\mathbb H}:=({\mathcal H}_t)_{t\ge 0}$.
 \item ${\mathcal P}({\mathbb F})$
(resp.\ ${\mathcal P}({\mathbb G})$): 
$\sigma$-algebra generated by 
${\mathbb F}$ (resp.\ ${\mathbb G}$)-predictable measurable subsets
on ${\mathbb R}_+ \times \Omega$.
Equivalently, $\sigma$-algebra on ${\mathbb R}_+ \times \Omega$
generated by ${\mathbb F}$-adapted left-continuous processes.  
 \item ${\mathcal O}({\mathbb F})$
(resp.\ ${\mathcal O}({\mathbb G})$):
$\sigma$-algebra generated by 
${\mathbb F}$ (resp.\ ${\mathbb G}$)-optional measurable subsets
on ${\mathbb R}_+ \times \Omega$.
Equivalently, $\sigma$-algebra on ${\mathbb R}_+ \times \Omega$
generated by ${\mathbb F}$-adapted right-continuous processes.  
 \item $\displaystyle {\mathcal P}_{{\mathbb F}}$
(resp.\ $\displaystyle {\mathcal P}_{{\mathbb G}}$): 
the space of ${\mathbb F}$ (resp.\ ${\mathbb G}$)-predictable 
processes. 
 \item $\displaystyle {\mathcal O}_{{\mathbb F}}$
(resp.\ $\displaystyle {\mathcal O}_{{\mathbb G}}$): 
the space of ${\mathbb F}$ (resp.\ ${\mathbb G}$)-optional 
processes. 
 \item ${\mathcal P}_{{\mathbb F}}^{(k)}$:
the space of the parametrized processes, 
$f: {\mathbb R}_+\times \Omega\times {\mathbb R}_{+}^k
\ni (t,\omega,u)\mapsto 
f_t(\omega,u)\in {\mathbb R}$, 
which is ${\mathcal P}({\mathbb F})\otimes {\mathcal B}({\mathbb R}_{+}^k)/
{\mathcal B}({\mathbb R})$-measurable.
 \item ${\mathcal O}_{{\mathbb F}}^{(k)}$:
the space of the parametrized processes, 
$f: {\mathbb R}_+\times \Omega\times {\mathbb R}_{+}^k
\ni (t,\omega,u)\mapsto 
f_t(\omega,u)\in {\mathbb R}$, 
which is ${\mathcal O}({\mathbb F})\otimes {\mathcal B}({\mathbb R}_{+}^k)/
{\mathcal B}({\mathbb R})$-measurable.
 \item 
Denote by 
$\displaystyle {\mathcal P}_{{\mathbb F},t}:=
\left\{ f 1_{[0,t]} | \ f\in {\mathcal P}_{\mathbb F}\right\}$, 
$\displaystyle {\mathcal O}_{{\mathbb F},t}:=
\left\{ f 1_{[0,t]} | \ f\in {\mathcal O}_{\mathbb F}\right\}$, 
$\displaystyle {\mathcal P}^{(k)}_{{\mathbb F},t}:=
\left\{ f(\cdot) 1_{[0,t]} | \ f\in {\mathcal P}^{(k)}_{\mathbb F}\right\}$, 
and $\displaystyle {\mathcal O}^{(k)}_{{\mathbb F},t}:=
\left\{ f(\cdot) 1_{[0,t]} | \ f\in {\mathcal O}^{(k)}_{\mathbb F}\right\}$, 
for example. 
\end{itemize}
\end{nota}
We recall the following basic properties 
of stochastic processes under 
the progressively enlarged filtration ${\mathbb G}$.
\begin{lem}[Lemmas 5.1 and 2.1 of Pham, 2010]
\quad
\begin{itemize}
 \item[{\rm (1)}]
Any ${\mathcal G}_t$-predictable process $\left(P(t)\right)_{t\ge 0}$
has the expression that
\begin{multline*}
 P(t) = 
p_0(t) 1_{\{t\le \tau_1\wedge\tau_2\}} 
+ p^1_t(\tau_1) 1_{\{\tau_1< t\le \tau_2\}}
+ p^2_t(\tau_2) 1_{\{\tau_2< t\le \tau_1\}}
+ p^{1,2}_t(\tau_1,\tau_2) 1_{\{t> \tau_1\vee\tau_2\}},
\end{multline*}
where $\left(p_0(t)\right)_{t\ge 0}\in {\mathcal P}_{\mathbb F}$,  
$\left(p^i_t(\cdot) \right)_{t\ge 0}\in {\mathcal P}_{\mathbb F}^{(1)}$
($i=1,2$)
and 
$\left(p^{1,2}_t(\cdot,\cdot) \right)_{t\ge 0}\in 
{\mathcal P}_{\mathbb F}^{(2)}$.
 \item[{\rm (2)}]
Any ${\mathcal G}_t$-optional process $\left(P(t)\right)_{t\ge 0}$
has the expression that
\begin{multline*}
 P(t) = 
p_0(t) 1_{\{t< \tau_1\wedge\tau_2\}} 
+ p^1_t(\tau_1) 1_{\{\tau_1\le t<\tau_2\}}
+ p^2_t(\tau_2) 1_{\{\tau_2\le t<\tau_1\}}
+ p^{1,2}_t(\tau_1,\tau_2) 1_{\{t\ge \tau_1\vee\tau_2\}},
\end{multline*}
where $\left(p_0(t)\right)_{t\ge 0}\in {\mathcal O}_{\mathbb F}$, 
$\left(p^i_t(\cdot) \right)_{t\ge 0}\in {\mathcal O}_{\mathbb F}^{(1)}$
($i=1,2$)
and 
$\left(p^{1,2}_t(\cdot,\cdot) \right)_{t\ge 0}\in 
{\mathcal O}_{\mathbb F}^{(2)}$.

 \item[{\rm (3)}]
Any ${\mathcal G}_t$-measurable random variable $G_t$
has the expression that
\begin{multline*}
G_t = 
g^0_t 1_{\{t< \tau_1\wedge\tau_2\}} 
+ g^1_t(\tau_1) 1_{\{\tau_1\le t<\tau_2\}}
+ g^2_t(\tau_2) 1_{\{\tau_2\le t<\tau_1\}}
+ g^{1,2}_t(\tau_1,\tau_2) 1_{\{t\ge \tau_1\vee\tau_2\}},
\end{multline*}
where $g^0_t$ is an ${\mathcal F}_t$-measurable random variable, 
$\left( g^i_t(\cdot) \right)_{t\ge 0}
\in {\mathcal O}_{\mathbb F}^{(1)}$ ($i=1,2$), and
$\left( g^{1,2}_t(\cdot,\cdot)\right)_{t\ge 0}
\in {\mathcal O}_{\mathbb F}^{(2)}$.
\end{itemize}
\end{lem}
Now, on the filtered probability space 
$(\Omega,{\mathcal F},{\mathbb P}, {\mathbb G})$, 
we consider the BSDE 
\begin{equation}
\begin{split}
 -dY(t) =& f\left( t, Y(t), Z(t), U_1(t),U_2(t) \right)dt\\
&- Z(t)^\top dW(t) - U_1(t) dM_1(t)-U_2(t) dM_2(t), \\
&\qquad t\in [0, \tau_1 \wedge \tau_2\wedge T], \\
Y({\tau_1 \wedge \tau_2\wedge T})=& 
\phi_1(\tau_1) 1_{\{ \tau_1<\tau_2\wedge T\}}
+\phi_2(\tau_2) 1_{\{ \tau_2<\tau_1\wedge T\}}
+\xi_T 1_{\{ T<\tau_1 \wedge \tau_2\}},
\end{split}
\end{equation}
where
$T\in {\mathbb R}_{++}:=(0,\infty)$
is a fixed terminal time, and the following conditions are imposed.
\begin{ass}
\begin{itemize}
\item[\rm (i)]
$\xi_T \in L^2(\Omega , {\mathcal F}_{T}, {\mathbb P})$. 
\item[\rm (ii)]
For $i=1,2$, 
$\phi_i \in {\mathcal O}_{\mathbb F}$
so that
$\displaystyle 
{\mathbb E}\left[ \sup_{t\in [0,T]}|\phi_i(t)|^2\right]<\infty$.
\item[\rm (iii)]	
$f: [0,T]\times \Omega \times {\mathbb R} \times {\mathbb R}^n
\times {\mathbb R}^2 \to {\mathbb R}$ is 
${\mathcal P}_{{\mathbb F}}\otimes {\mathcal B}({\mathbb R})
\otimes {\mathcal B}({\mathbb R}^n)\otimes {\mathcal B}({\mathbb R}^2)
/{\mathcal B}({\mathbb R})$-measurable and satisfies, 
with some positive constant $K_f>0$, 
\begin{multline*}
\left|f \left(t,y, z, u_1, u_2\right) 
-f \left(t,y', z', u'_1, u'_2\right) \right| \\
\le K_f
\left(|y-y'|+|z-z'|+|u_1-u_1'|+|u_2-u_2'|\right) \\
\text{for all $(y, z, u_1, u_2)$, $(y',z',u_1', u'_2)$}
\end{multline*}
a.e.\ $(t,\omega)\in [0,T]\times \Omega$. 
\item[\rm (iv)]
It holds that
\[
{\mathbb E}
\left[\int_{0}^{T} \left|f(t,0,0,0,0)\right|^2
dt\right]< \infty.
\]
\end{itemize} 
\end{ass}
\subsection{Existence, Uniqueness, and Construction of Solution}

A specific feature of BSDE (3)
is that it has the random time horizon $\tau_1\wedge \tau_2 \wedge T$, 
where $\tau_i$ is the (first) jump time for the martingale $M_i$
($i=1,2$).
As for the definition of the solution to such a BSDE, 
we employ the following
(cf.\ Darling and Pardoux, 1997 as an example of related work).
\begin{defi}
We call the quadruplet $(Y,Z,U^1,U^2): 
[0,T]\times\Omega \to {\mathbb R}\times {\mathbb R}^n \times 
{\mathbb R}\times {\mathbb R}$ a solution to BSDE (3) if it satisfies
the following conditions.
\begin{itemize}
 \item[\rm (a)] 
$Y:=(Y(t))_{t\in [0,T]}$ is a ${\mathbb G}$-adapted RCLL 
(i.e., right continuous and having left limit) process 
(which is an element of ${\mathcal O}_{{\mathbb G},T}$), and 
$(Z,U^1,U^2)\in \left( {\mathcal P}_{{\mathbb G},T}\right)^{n+2}$.
 \item[\rm (b)]
For $t\in [0,T]$, it holds that
\begin{align*}
Y(t)
1_{\{ \tau_1\wedge \tau_2 \le t\}}
=&
\left\{
\phi_1({\tau_1}) 1_{\{ \tau_1<\tau_2\}}
+\phi_2({\tau_2}) 1_{\{ \tau_2<\tau_1\}}
\right\}1_{\{ \tau_1\wedge \tau_2\le t\}}, \\
Z(t)1_{\{ \tau_1\wedge \tau_2 \le t\}}=&0, \\
U_i(t)1_{\{ \tau_1\wedge \tau_2 \le t\}}=&0, \quad i=1,2.
\end{align*}
 \item[\rm (c)] For $t\in [0,T]$, it holds that
\[
\begin{split}
Y(t)=&\phi_1({\tau_1}) 1_{\{ \tau_1<\tau_2, \tau_1 \le T\}}
+\phi_2({\tau_2}) 1_{\{ \tau_2<\tau_1, \tau_2 \le T\}}
+\xi_T 1_{\{ \tau_1 \wedge \tau_2 >T\}} \\
&+
\int_{t\wedge \tau_1\wedge\tau_2}^{T\wedge \tau_1\wedge\tau_2} 
f\left( s, Y(s), Z(s), U_1(s),U_2(s) \right)ds \\
&- \int_{t\wedge \tau_1\wedge\tau_2}^{T\wedge \tau_1\wedge\tau_2} 
\left\{
Z(s)^\top dW(s)  +U_1(s) dM_1(s)+U_2(s) dM_2(s)
\right\}.
\end{split}
\]
\end{itemize}
\end{defi}
Furthermore, we define the following spaces of stochastic processes, namely, 
\begin{align*} 
{\mathbb S}^2_{\beta,T}:=&
\left\{ Y\in {\mathcal O}_{{\mathbb G},T} 
\bigm| \| Y\|^2_{\beta,T}<\infty
\right\}, \\
{\mathbb H}^{2,d}_{\beta,T}:=& 
\left\{ Z\in \left( {\mathcal P}_{{\mathbb G},T}\right)^d
\bigm| \|Z\|^2_{\beta,T}<\infty\right\},
\end{align*} 
letting $\beta\in {\mathbb R}$ and denoting 
\[
\| Y\|^2_{\beta,T}:=
{\mathbb E}\left[ \int_0^T  e^{\beta t}|Y(t)|^2dt\right].
\]
We then obtain the following.
\begin{thm}
Under Assumption 1, BSDE (3) admits a  
unique solution \\ 
$(Y, Z, U_1, U_2) \in 
{\mathbb S}^2_{\beta,T} \times {\mathbb H}^{2,n+2}_{\beta,T}$
for any sufficiently large $\beta>0$.
\end{thm} 
 
\begin{proof}[Sketch]
The method of proof is standard, 
although the horizon is random, which is rather ``non-standard''.
We consider a Picard-type iteration, that is, 
for a given $\left(\bar{Y}, \bar{Z}, \bar{U}^1, \bar{U}^2\right) 
\in {\mathbb S}^2_{\beta,T} \times {\mathbb H}^{2,n+2}_{\beta,T}$,
we construct the solution to BSDE
\begin{equation}
\begin{split}
-dY(t)=&f \left(t,\bar{Y}(t), \bar{Z}(t),
\bar{U}_1(t), \bar{U}_2(t)\right)dt \\
&- Z(t)^\top d W(t)- U_1(t) d M_1(t) - U_2(t) d M_2(t), \\
&\qquad t\in [0,\tau], \\
Y(\tau) =&\zeta,
\end{split}
\end{equation}
where we denote 
\begin{align*}
\tau_0:=& \tau_1\wedge\tau_2, \quad
\tau:=\tau_0\wedge T, \\
\zeta:=&
\phi_1({\tau_1}) 1_{\{ \tau_1<\tau_2\wedge T\}}
+\phi_2({\tau_2}) 1_{\{ \tau_2<\tau_1\wedge T\}}
+\xi_T 1_{\{ T<\tau_1 \wedge \tau_2 \}}.
\end{align*}
Indeed, using the ${\mathbb G}$-martingale representation 
\begin{align*} 
{\mathcal M}(t):=&{\mathbb E}\left[ \zeta 
+ \int_{0}^{\tau}
f \left(u,\bar{Y}(u), \bar{Z}(u), \bar{U}_1(u), \bar{U}_2(u) \right)  
du \biggm|{\mathcal G}_t \right] \\
=&{\mathbb E}\left[ \zeta 
+ \int_{0}^{\tau}
f \left(u,\bar{Y}(u), \bar{Z}(u), \bar{U}_1(u), \bar{U}_2(u) \right)  
du \right] \\
&+\int_0^{t} \phi(u)^\top d W(u)
+ \int_0^{t} \psi_1(u) d M_1(u)
+ \int_0^{t} \psi_2(u) d M_2(u),
\quad t\in [0,T]
\end{align*}  
for some $(\phi,\psi^1,\psi^2)\in {\mathbb H}^{2,n+2}_{\beta,T}$
(e.g., see Section~5.2 of Bielecki and Rutkowski, 2004), 
we define 
\begin{align*}
\tilde{Y}_t:=& {\mathbb E}\left[ \zeta 
+ \int_{t\wedge \tau}^{\tau}
f \left(u,\bar{Y}_u, \bar{Z}_u, \bar{U}_u^1, \bar{U}_u^2 \right)  
du \biggm|{\mathcal G}_t \right],
\quad t\in [0,T], \\
\tilde{Z}:\equiv& \phi,
\quad 
\tilde{U}^1:\equiv \psi^1,
\quad
\tilde{U}^2:\equiv \psi^2.
\end{align*}
Note that the martingale $({\mathcal M}_t)_{t\in [0,T]}$ 
with respect to the right-continuous
filtration ${\mathbb G}$ admits an RCLL modification. 
Hence, 
\[
 \tilde{Y}(t) = {\mathcal M}(t) -
\int_0^{t\wedge \tau}
f \left(u,\bar{Y}(u), \bar{Z}(u), \bar{U}_1(u), \bar{U}_2(u) \right)  
du
\]
also admits an RCLL modification, which is denoted by
$\left(\tilde{Y}(t)\right)_{t\in [0,T]}$ again. 
Furthermore, we can check the integrability,  
$\tilde{Y} \in {\mathbb S}^2_{\beta,T}$.
Hence, 
$\bigl(\tilde{Y},\tilde{Z},\tilde{U}_1,\tilde{U}_2\bigr)$
is the solution to (4).
Next, we show that the map
\[
\Psi: 
{\mathbb S}^2_{\beta,T} \times {\mathbb H}^{2,n+2}_{\beta,T} 
\ni \left( \bar{Y},\bar{Z},\bar{U}_1,\bar{U}_2\right)
\mapsto 
\left( \tilde{Y},\tilde{Z},\tilde{U}_1,\tilde{U}_2\right)
\in {\mathbb S}^2_{\beta,T} \times {\mathbb H}^{2,n+2}_{\beta,T} 
\]
is a contraction for sufficiently large $\beta>0$, 
and 
using the fixed point theorem for the contraction map,
we conclude that the fixed point of the map $\Psi$ is the solution. 
\end{proof} 
\begin{rem}
We refer to Section~19 of Cohen and Elliott (2015) for the detail
of such a Picard-type iteration argument, 
where a more general semimartingale BSDE (driven by L\'evy noise)
is treated with a {\it fixed constant} time horizon. 
\end{rem}
Actually, we can construct 
the solution to BSDE (3) on the filtered probability space 
$(\Omega,{\mathcal F},{\mathbb P}, {\mathbb G})$, 
using another reduced BSDE on the smaller filtered probability space 
$(\Omega,{\mathcal F},{\mathbb P}, {\mathbb F})$.
Assuming
\begin{ass}
$h_i$ ($i=1,2$) are bounded,
\end{ass}
we obtain the following.
\begin{thm}
Under Assumptions 1 and 2, 
the solution $(Y,Z,U_1,U_2)
\in {\mathbb S}^2_{\beta,T} \times {\mathbb H}^{2,n+2}_{\beta,T}$
has the representation that
\begin{equation}
\begin{split}
Y(t)=&
\bar{Y}(t)1_{\{ 0\le t<\tau_1\wedge\tau_2\wedge T \}} \\ 
+ \Bigl\{
\phi_1&({\tau_1}) 1_{\{ \tau_1<\tau_2\wedge T\}}
+\phi_2({\tau_2}) 1_{\{ \tau_2<\tau_1\wedge T\}}
+\xi_T 1_{\{ T<\tau_1 \wedge \tau_2\}}
\Bigr\}
1_{\{ t=\tau_1 \wedge \tau_2\wedge T\}}, \\
Z(t)=& \bar{Z}(t), \\
U_i(t)=& \phi_i(t)-\bar{Y}(t), \quad i=1,2.
\end{split}
\end{equation}
Here,  
$\left( \bar{Y},\bar{Z}\right)
\in {\mathbb S}^2_{\beta,T} \times {\mathbb H}^{2,n}_{\beta,T}$ 
is the solution to a BSDE
on $(\Omega,{\mathcal F},{\mathbb P}, {\mathbb F})$, namely,
\begin{equation}
\begin{split}
 -d\bar{Y}(t) =& \bar{f}\left( t, \bar{Y}(t), \bar{Z}(t)
\right)dt
- \bar{Z}(t)^\top dW(t), 
\quad t\in [0,T], \\
Y_T=&\xi_T,
\end{split}
\end{equation}
where
\begin{align*}
 \bar{f}(t,y,z):=f\left(t,y,z,\phi_1(t)-y,\phi_2(t)-y\right) 
+ \left\{ \phi_1(t)-y\right\} h_1(t)
+\left\{ \phi_2(t)-y\right\} h_2(t). 
\end{align*}
\end{thm}

\begin{rem}
Similar reduction results for BSDEs (into smaller filtrations)
have been studied by 
Cr\'epey and Song (2016) and Pham (2010) in more-general settings.
\end{rem}

\begin{proof}[Sketch]
Note that BSDE (3) is rewritten as
\begin{equation}
\begin{split}
 -dY(t) =& \tilde{f}\left( t, Y(t), Z(t), U_1(t),U_2(t) \right)dt
- Z(t)^\top dW(t) \\
&\text{on \ $\{ 0\le t<\tau_1 \wedge \tau_2\wedge T \}$,} \\
{\mathit\Delta}Y(t)=&
U_1({\tau_1}) 1_{\{ \tau_1<\tau_2 \wedge  T\}}
+U_2({\tau_2}) 1_{\{ \tau_2<\tau_1 \wedge T\}}, \\
Y(t)=& 
\phi_1({\tau_1}) 1_{\{ \tau_1<\tau_2\wedge T\}}
+\phi_2({\tau_2}) 1_{\{ \tau_2<\tau_1\wedge T\}}
+F_T 1_{\{ T<\tau_1 \wedge \tau_2\}}\\
&\text{on \ $\{ t={\tau_1 \wedge \tau_2\wedge T}\}$,}
\end{split}
\end{equation}
where we use 
${\mathit\Delta}Y(t):=Y(t)-Y(t-)$ and
\[
 \tilde{f}\left(t,y,z,u_1,u_2\right)
=f\left(t,y,z,u_1,u_2\right)
+u_1 h_1(t)+ u_2 h_2(t).
\]
We show that if we define $(Y,Z,U^1,U^2)$ by (5), then 
it actually satisfies (7).
First, we see that BSDE (6)
on $(\Omega,{\mathcal F},{\mathbb P}, {\mathbb F})$
has a unique solution $(\bar{Y},\bar{Z})\in {\mathbb S}^2_{\beta,T}
\times {\mathbb H}^{2,n}_{\beta,T}$ for any sufficiently large 
$\beta>0$, recalling that
$\bar{f}$ is a standard driver 
(e.g., $\bar{f}(t,y,z)$ satisfies a globally Lipschitz condition
with respect to $(y,z)$).
Next, we can check that (5) indeed satisfies (7); for example, 
on $\{ t=\tau_1\wedge\tau_2\wedge T\}$, 
\begin{align*}
{\mathit\Delta}{Y}(t) =&
\phi_1({\tau_1}) 1_{\{ \tau_1<\tau_2\wedge T\}}
+\phi_2({\tau_2}) 1_{\{ \tau_2<\tau_1\wedge T\}}
+\xi_T 1_{\{ T<\tau_1 \wedge \tau_2\}}
- \bar{Y}(t-)  \\
=&
\phi_1({\tau_1}) 1_{\{ \tau_1<\tau_2\wedge T\}}
+\phi_2({\tau_2}) 1_{\{ \tau_2<\tau_1\wedge T\}}
+\xi_T 1_{\{ T<\tau_1 \wedge \tau_2\}} \\
&-\left( \bar{Y}(\tau_1\wedge\tau_2) 
1_{\{\tau_1\wedge\tau_2\le T\}}+\xi_T 1_{\{\tau_1\wedge\tau_2>T\}}\right) \\
=& 
U_1({\tau_1}) 1_{\{ \tau_1<\tau_2 \wedge  T\}}
+U_2({\tau_2}) 1_{\{ \tau_2<\tau_1 \wedge T\}}.
\end{align*}
Hence, the desired assertion follows as it is easy to see
the integrabilities given by (5), $(Y,Z,U_1,U_2)
\in {\mathbb S}^2_{\beta,T} \times {\mathbb H}^{2,n+2}_{\beta,T}$.
\end{proof}

\begin{rem}
We impose Assumption~2 to simplify the statement of Theorem~2. 
We can relax it by employing a different solution space 
(from ${\mathbb S}^2_{\beta,T}\times {\mathbb H}^{2,n+2}_{\beta,T}$) 
associated with the so-called {\it stochastic} Lipschitz BSDEs.
For the study of such BSDEs, 
see El Karoui and Huang (1997) and Nagayama (2019), for example.
\end{rem}

\subsection{Markovian Model}

When we treat BSDE (3) in a practical application, 
more-concrete modeling is preferable:
In this subsection, 
we consider BSDE (3) under Assumptions~1 and 2 and the following setting.
\begin{itemize}
 \item[(i)] 
There is a Markovian state variable process $X:=(X(t))_{t\ge 0}$, 
which is governed by the following Markovian forward stochastic 
differential equation (FSDE), namely, 
\begin{equation}
dX(t) = b(t,X(t))dt + a(t,X(t))dW(t),
\quad X(0)\in {\mathbb  R}^d,
\end{equation}
on $(\Omega,{\mathcal F},{\mathbb P}, ({\mathcal F}_t)_{t\ge 0})$,
where $a:{\mathbb R}_+ \times{\mathbb R}^d \to {\mathbb R}^{d\times n}$
and $b:{\mathbb R}_+ \times {\mathbb R}^d \to {\mathbb R}^{d}$. 
 \item[(ii)]
$h_i(t):=\tilde{h}_i(X(t))$, $i=1,2$, 
where $\tilde{h}_i:{\mathbb R}^d \to {\mathbb R}_+$ is bounded.
 \item[(iii)]
The driver 
$f: [0,T]\times \Omega \times {\mathbb R} \times {\mathbb R}^n
\times {\mathbb R}^2 \to {\mathbb R}$
of BSDE (3) is written as
\[
 f(t,\omega, y,z,u_1,u_2):=
g(t,X(t,\omega),y,z,u_1,u_2),
\]
where $g:[0,T]\times {\mathbb R}^d
\times {\mathbb R}\times {\mathbb R}^n
\times {\mathbb R}\times {\mathbb R}\to {\mathbb R}$.
 \item[(iv)]
$\displaystyle \xi_T:=\Xi(X(T))$, 
where $\Xi:{\mathbb R}^d\to {\mathbb R}$.
 \item[(v)]
$\displaystyle \phi_i(t):=\varphi_i(X(t))$, $i=1,2$, 
where $\varphi_i:{\mathbb R}^d\to {\mathbb R}$.
\end{itemize}
In this case, the solution to BSDE (3) can be constructed as follows 
using the solution to a second-order parabolic semilinear partial differential equation (PDE).
\begin{thm}
Consider the second-order parabolic semilinear PDE 
\begin{equation}
\begin{split}
-\partial_t V(t,x)=&{\mathcal L}_t V(t,x)
+ \bar{g}\left( t,x, V(t,x), a(t,x)^\top \nabla V(t,x)\right), 
\quad (t,x)\in [0,T)\times {\mathbb R}^d, \\
V(T,x)=&\Xi(x),
\end{split}
\end{equation}
where 
\begin{equation}
 {\mathcal L}_tV:=\frac{1}{2}
{\rm tr}\left( aa^\top(t,\cdot) \nabla\nabla V\right)
+ b^\top(t,\cdot) \nabla V
\end{equation}
is the infinitesimal generator for $X$ 
with the gradient 
$\nabla V:=\left( \partial_{x_1} V,\dots,\partial_{x_d} V\right)^\top$
and the Hessian matrix
$\nabla\nabla V:=\left( \partial_{x_i x_j}^2 V\right)_{1\le i,j\le d}$, 
and
\[
\bar{g} (t,x,y,z):= 
g\left( t,x, y, z, \varphi_1(x)-y, \varphi_2(x)-z\right) 
+ \sum_{i=1}^2 \left\{ \varphi_i(x)-y\right\}\tilde{h}_i(x).
\]
Suppose that there exists a unique classical solution 
$V\in C^{1,2}([0,T]\times {\mathbb R}^d)$ to (9). 
Then, the solution to BSDE (3) is represented as 
\begin{align*}
Y(t)=&
V\left(t,X(t)\right)1_{\{ 0\le t<\tau_1\wedge\tau_2\wedge T\}} 
+ \Bigl\{
\varphi_1\left(X({\tau_1})\right) 1_{\{ \tau_1<\tau_2 \wedge T\}} \\
&+\varphi_2\left(X({\tau_2})\right) 1_{\{ \tau_2<\tau_1 \wedge T\}}
+\Xi\left(X(T)\right) 1_{\{ T<\tau_1 \wedge \tau_2\}}
\Bigr\} 1_{\{ t=\tau_1 \wedge \tau_2\wedge T\}}, \\
Z(t)=& a\left(t,X(t)\right)^\top \nabla V\left(t,X(t)\right), \\
U_i(t)=& \varphi_i\left(X(t)\right)-V\left(t,X(t)\right), \quad i=1,2.
\end{align*}
\end{thm}
\begin{proof}[Sketch]
Associated with BSDE (6), we consider the (decoupled) forward-backward stochastic differential equation (FBSDE)
\begin{equation}
\begin{split}
dX(t)=& b\left( t,X(t)\right)dt + a\left( t,X(t)\right)dW(t), \\
X(0)\in& {\mathbb R}^d, \\
-d\bar{Y}(t) =&
\bar{g}\left( t,X(t), \bar{Y}(t), \bar{Z}(t)\right) dt
-\bar{Z}(t)^\top dW(t), \\
\bar{Y}(T)=&\Xi(X(T)).
\end{split}
\end{equation}
By the nonlinear Feynman--Kac formula 
(e.g., see El Karoui et al., 2000 or Zhang, 2017), 
the solution to (11) is expressed as 
\[
 \bar{Y}(t):=V\left(t,X(t)\right),
\quad 
\bar{Z}(t):=a\left(t,X(t)\right)^\top\nabla V\left(t,X(t)\right),
\quad t\in [0,T].
\]
The desired assertion follows by using Theorem~2.
\end{proof}

\begin{rem}
In the study of credit risk modeling in mathematical finance, 
similar techniques, namely 
the reduction of a BSDE (onto a Brownian filtration)
combined with the (nonlinear) Feynman--Kac formula,  
have been utilized: 
see Bichuch et al.\ (2015), Bielecki et al.\ (2005), 
and Cr\'epey (2015), for example. 
\end{rem}

\section{XVA Calculation via BSDE}

In this section, we introduce a 
``post-crisis'' financial market model
and a hedger's model for pricing OTC financial derivative securities, 
which generalize those employed by Bichuch et al.\ (2015, 2018) 
and Tanaka (2019).
We then derive BSDEs that describe the self-financing hedging 
portfolio values of the hedger (seller) and her counterparty (buyer). 
After preparing mathematical models of a financial market, 
a hedger, and her counterparty, 
we formulate hedging problems and give the definition 
of the arbitrage-free price of a derivative security.
Throughout this section, 
we continue to use the mathematical setup introduced in Section~2.

\subsection{Non-defaultable/Defaultable Risky Assets}

Let $T\in {\mathbb R}_{++}$ be a fixed time horizon, and 
consider a frictionless financial market model in continuous time. 
In it, there are price processes of
$n$ non-defaultable risky assets 
$S:=(S_1,\dots,S_n)^\top$, $S_i:=(S_i(t))_{t\in [0,T]}$, 
one defaultable risky asset 
$P_I:=(P_I(t))_{t\in [0,T]}$ issued by an investor's firm, 
and one defaultable risky asset 
$P_C:=(P_C(t))_{t\in [0,T]}$ issued by the firm 
of a counterparty of the investor. 
They are governed by the following stochastic differential equations (SDEs)
on $(\Omega,{\mathcal F},{\mathbb P}, {\mathbb G})$: 
\begin{align}
dS(t) =& \text{diag}\left(S(t)\right) 
\left\{ \sigma(t) dW(t) + r_D(t){\bf 1}dt\right\}, 
\quad S(0)\in {\mathbb R}^n_{++}, \\
dP_I(t) =& P_I(t-) 
\left\{ \sigma_I(t)dW(t) -dM_1(t)+r_D(t) dt\right\}, 
\quad P_I(0)\in {\mathbb R}_{++},\\
dP_C(t) =& P_C(t-) 
\left\{ \sigma_C(t)dW(t)-dM_2(t)+r_D(t) dt\right\},
\quad P_C(0)\in {\mathbb R}_{++}.
\end{align}
Here,  
$\sigma\in ({\mathcal P}_{{\mathbb F},T})^{n\times n}$,
$\sigma_i\in ({\mathcal P}_{{\mathbb F},T})^{1\times n}$, $i\in \{I,C\}$, 
and $r_D\in {\mathcal P}_{{\mathbb F},T}$, 
which are assumed to be bounded, 
and $\sigma(t,\omega)$ is invertible 
for a.e.\ $(t,\omega)\in [0,T]\times\Omega$. 
Furthermore, we denote 
$\text{diag}(x)=(x_i \delta_{ij})_{1\le i,j\le n}$ 
for $x:=(x_1,\dots,x_n)^\top \in {\mathbb R}^n$ 
and ${\bf 1}:=(1,\dots,1)^\top \in {\mathbb R}^n$. 
\begin{rem}
We regard the process $r_D$ as the risk-free interest rate process
in the market, which does not contain credit risk spread.\footnote{
A typical example of such an interest rate in a real financial market
is the OIS rate.}
Define the cash account process 
$B_D:=(B_D(t))_{t\ge 0}$ 
associated with the risk-free rate $r_D$ by
\[
 dB_D(t)=B_D(t) r_D(t)dt,
\quad B_D(0)=1,
\]
or equivalently 
\[
 B_D(t)= \exp\left\{ \int_0^t r_D(u)du\right\}.
\]
We then see that
\[
\frac{S_i}{B_D},
\quad
i=1,\dots,n,
\quad
\frac{P_j}{B_D}, 
\quad j=1,2
\]
are $\mathbb G$-local martingales.
These mean that we are starting with the probability space
$(\Omega,{\mathcal F},{\mathbb P})$ with 
a risk-neutral (pricing) probability $\mathbb P$,\footnote{More
 precisely, $\mathbb P$ is an equivalent martingale measure (EMM). 
See Remark~13 in Subsection~3.5.} 
not with the real-world (physical) probability. 
\end{rem}
The random times $\tau_1$ and $\tau_2$ defined by (2)
are interpreted as the default times of the investor who issues $P_I$
and the counterparty who issues $P_C$, respectively.
We solve (13) as
\begin{align*}
 P_I(t)=P_I(0) 
 \exp\left[
\int_0^t \sigma_I(u)dW(u)
+\int_0^t\left(
r_D(u)+ h_1(u)
-\frac{1}{2} |\sigma_I(u)|^2
\right)du
\right]
\left\{ 1-N_1(t)\right\},
\end{align*}
for example.
Recall that the price becomes zero when defaults occur, 
i.e., $P_I(\tau_1)=0$.
\begin{rem}
As concrete examples of $P_I$  and $P_C$,
we can consider $T$-maturity zero coupon bonds
without recoveries, namely 
\begin{align*}
 P_I(t)=& {\mathbb E}\left[ \exp\left\{ -\int_t^T 
\left( r_D(u)+h_1(u)\right)du
\right\}\biggm| {\mathcal F}_t\right]
\left\{ 1-N_1(t)\right\}, \\
 P_C(t)=& {\mathbb E}\left[ \exp\left\{ -\int_t^T 
\left( r_D(u)+h_2(u)\right)du
\right\}\biggm| {\mathcal F}_t\right]
\left\{ 1-N_2(t)\right\}. 
\end{align*}
The volatility terms 
$(\sigma_j(t))_{t\in [0,T]}$ ($j\in\{I,C\}$) are described 
by using 
the $({\mathbb P},{\mathcal F}_t)$-Brownian 
martingale representation: 
For example, in the $j=I$ case, 
$\left( \sigma_I(t)\right)_{t\in [0,T]}$
is determined to satisfy
\begin{align*}
{\mathbb E}\left[ \exp\left\{ -\int_0^T 
\left(r_D(u)+h_1(u)\right)du\right\}
\biggm|{\mathcal F}_t\right]  
=P_I(0) \exp\left\{
\int_0^t \sigma_I(s)dW(s)
-\frac{1}{2}\int_0^t |\sigma_I(s)|^2 ds
\right\}
\end{align*}
for $ t\in [0,T]$.
\end{rem}

\subsection{Defaultable Derivative Security}

We treat the following derivative security
in our financial market model.
\begin{defi}
A European derivative security is described as
\[
\left( T, \tau_1,\tau_2, \xi_T, \phi_1,\phi_2\right), 
\]
where 
$\xi_T\in L^2\left( \Omega,{\mathcal F}_T,{\mathbb P}\right)$
and
$\phi_i\in \left\{ \phi\in {\mathcal O}_{{\mathbb F},T} \bigm| 
{\mathbb E}\left[ \sup_{t\in [0,T]}|\phi(t)|^2 \right]<\infty\right\}$ 
$(i=1, 2)$.
Here, 
\begin{itemize}
 \item $\tau_1\wedge \tau_2\wedge T$ is the maturity,
 \item 
$\xi_T$ is the payoff at the maturity when no default occurs,
 \item $\phi_1(\tau_1)$ is the payoff at the maturity when the investor
       defaults,
 \item $\phi_2(\tau_2)$ is the payoff at the maturity when the 
counterparty defaults.
\end{itemize}
This means that at the maturity, 
\begin{equation}
H:=
 \xi_T 1_{\{ T<\tau_1\wedge\tau_2\}}
+\phi_1(\tau_1) 1_{\{ \tau_1 <\tau_2, \tau_1\le T\}}
+\phi_2(\tau_2) 1_{\{ \tau_2 <\tau_1, \tau_2\le T\}}
\end{equation}
is paid to the counterparty (buyer)
from the investor (seller, writer).
\end{defi}
\begin{rem}
A typical example of the payoff $(\xi_T,\phi_1,\phi_2)$ is 
\[
\xi_T :=h\left((S(t))_{t\in [0,T]}\right)
\]
with $h: C([0,T],{\mathbb R}^n_{++}) \to {\mathbb R}$ and, 
for $i=1,2$, 
\[
 \phi_i(t):=\varphi_i\left( \hat{V}(t)\right)
\]
with some nonlinear (piecewise-linear) 
$\varphi_i:{\mathbb R}\to {\mathbb R}$ and
\begin{equation}
\hat{V} (t):=
{\mathbb E}\left[\exp\left\{-\int_t^T r_D(u)du\right\} \xi_T
\biggm| {\mathcal F}_t\right],
\quad t\in [0,T].
\end{equation}
(16) is interpreted as the reference value process of 
the derivative $(T,\xi_T)$ with the payoff $\xi_T$ at the maturity $T$
in a default-free market.
In Bichuch et al.\ (2018), 
\begin{equation}
\varphi_1(v):= 
v -L_I \left( v-\alpha v\right)^{+} 
\quad\text{and}\quad
\varphi_2(v):= 
v +L_C \left( v-\alpha v\right)^{-}
\end{equation}
are employed, 
where $x^+:=\max(x,0)$, $x^-:=\max(-x,0)=-\min(x,0)$, 
$0\le L_I,L_C,\alpha\le 1$.
The constant $L_I$ (resp.\ $L_C$) is called the loss rate upon default of 
the investor (resp.\ the counterparty), 
and $\alpha$ is called the  collateralization level.
For a more detailed explanation, 
see Sections~3.2 and 3.4 of Bichuch et al.\ (2018).
\end{rem}

\subsection{Dynamic Portfolio Strategy}

For hedging purposes,
the writer (seller) of the derivative security given in Definition~2
constructs a dynamic portfolio, which is 
denoted by $\left( \pi,\pi^I,\pi^C,\pi^f,\pi^r,\pi^{col}\right)$. 
Here, 
\[
\pi:=\left( \pi_1,\dots,\pi_n\right)^\top
\in \left( {\mathcal P}_{{\mathbb G},T}\right)^n,
\quad
\pi_j:=(\pi_j(t))_{t\in [0,T]}
\]
is an investment strategy for the risky assets $S:=(S^1,\dots,S^n)^\top$,
\[
\pi^j:=(\pi^j(t))_{t\in [0,T]} \in {\mathcal P}_{{\mathbb G},T},
\quad 
j\in \{ I,C\}
\]
are investment strategies for the risky assets $P_I$ and $P_C$, 
respectively, and
\[
\pi^j:=(\pi^j(t))_{t\in [0,T]} \in {\mathcal P}_{{\mathbb G},T},
\quad 
j\in \{ f,r,{col}\}
\]
are investment strategies for the cash accounts 
$B_f$, $B_r$, and $B_{col}$, 
which are called 
the funding account,
the repo account, and
the collateral account, respectively.
They are defined by 
\begin{equation}
 dB_j(t) = B_j(t)
\left\{ r_j^-(t) 1_{\{ \pi^j(t)<0\}} 
+r_j^+ (t) 1_{\{ \pi^j(t)>0\}}\right\} dt,
\quad B_j(0)=1
\end{equation}
with 
$r_j^{-}:=(r_j^{-}(t))_{t\in [0,T]}\in {\mathcal P}_{{\mathbb F},T}$, 
$r_j^{+}:=(r_j^{+}(t))_{t\in [0,T]}\in {\mathcal P}_{{\mathbb F},T}$, 
and $j\in\{ f,r,col\}$, 
where $r_f^\pm, r_r^\pm$ and $r_{col}^\pm$ are called
the funding rate, the repo rate, and the collateral rate, respectively.
\begin{rem}
The cash account process $B_f$
represents the cumulative amount of cash
that the hedger borrows from (or lends to) her treasury desk.
The rate $r_f^-$ is called the funding borrowing rate
and the rate $r_f^+$ is called the funding lending rate. 
The cash account process $B_r$
represents the cumulative amount of cash
that the investor borrows from (or lends to) a repo market.
The rate $r_r^-$
is called the repo borrowing rate, which is applied 
when the hedger borrows money from the repo market and 
implements a long position
for the non-defaultable risky assets $S$. 
The rate $r_r^+$ is called the repo lending rate, which is applied
when the hedger lends money to the repo market and
implements a short-selling position
for the non-defaultable risky assets $S$. 
The cash account process $B_{col}$
represents the cumulative amount of cash
that the investor receives from (or posts to)
the counterparty as the collateral of the derivative security.
The rate $r_{col}^-$ is paid by the hedger to the counterparty if he/she 
has received the collateral.
The rates $r_{col}^+$ is received by the hedger 
if he/she has posted the collateral.
These rates can differ because different markets\footnote{
For example, the choice of currency (USD, Euro, etc.).
We refer the interested reader to Fujii and Takahashi (2011), 
where the impact of the choice of currency of collateral is studied.}
may be used to determine the contractual rates earned by cash collateral.
\end{rem}
For $r_f^{\pm}$ and $r_r^{\pm}$, it is natural and realistic 
to assume that
\begin{equation}
2\epsilon_j:\equiv r_j^--r_j^+\ge 0
\quad\text{for $j\in\{ f,r\}$.}
\end{equation}
For $j\in\{f,r\}$, denoting the ``mid-rate'' by
\[
r_j^0:\equiv \frac{r_j^-+r_j^+}{2},
\]
we see that
\[
r_j^{\pm}\equiv r_j^0 \mp \epsilon_j.
\]
The value process $Y:=(Y(t))_{t\in [0,T]}$ 
associated with a given dynamic portfolio strategy 
$\left( \pi,\pi^I,\pi^C,\pi^f,\pi^r,\pi^{col}\right)$
is governed by an SDE on $(\Omega,{\mathcal F},{\mathbb P},{\mathbb G})$, namely, 
\begin{equation}
\begin{split}
 dY(t) =& \pi(t)^\top dS(t)
+\pi^I(t) dP_I(t) + \pi^C(t) dP_C(t) \\
&+\pi^f(t) dB_f(t) + \pi^r(t) dB_r(t) 
+\pi^{col}(t) dB_{col}(t), \\
 Y(0)=&y, 
\end{split}
\end{equation}
subject to 
\begin{align}
&Y(t) = \pi(t)^\top S(t)
+\pi^I(t) P_I(t) + \pi^C(t) P_C(t) \nonumber \\
&\qquad\qquad+\pi^f(t) B_f(t) + \pi^r(t) B_r(t) 
+\pi^{\rm col}(t) B_{col}(t), \\
&\pi(t)^\top S(t) +\pi^r(t) B_r(t) =0, \\
&\pi^{col}(t) B_{col}(t)-\alpha \hat{V}(t)=0.
\end{align}
Here, 
(21) corresponds to the so-called self-financing condition,
(22) implies that the hedger accesses the repo market to purchase/sell
non-defaultable risky assets (stocks), 
and (23) implies that $\alpha \hat{V}(t)$ is regarded as 
the collateral value at time $t$, where
$\alpha\in [0,1]$ is the collateral level, which is 
the same as the one given in Remark~7.
From (21)--(23), recall that the relations 
\begin{align}
\pi^r(t)  =&-B_r(t)^{-1}\pi(t)^\top S(t), \\
\pi^{col}(t) =&B_{col}(t)^{-1}\alpha \hat{V}(t), \\
\pi^f(t)  =&
B_f(t)^{-1}\left\{
Y(t-) -\pi^I(t) P_I(t-) - \pi^C(t) P_C(t-)-\alpha \hat{V}(t)\right\}
\end{align}
hold. Hence, 
we can interpret that
$(y,\Pi)\in {\mathbb R} \times 
({\mathcal P}_{{\mathbb G},T})^{n+2}$, 
where $\Pi:=\left(\pi,\pi^I,\pi^C \right)$,  
is a portfolio strategy that determines the portfolio value process (20), 
and we sometimes write 
\[
 Y:\equiv Y^{(y,\Pi)},
\]
emphasizing the portfolio strategy $(y,\Pi)$. 
Combining (20) with (12)--(14), (18), and (24)--(26), we see that
\begin{align}
 dY(t) =& \pi(t)^\top \text{diag}\left(S(t)\right)
\left[ \sigma(t)dW(t) + 
\left\{ r_D(t)-r_r(t;\pi^r(t))\right\}{\bf 1}dt\right]
\nonumber \\
&+\pi^I(t) P_I(t-)\left[ \sigma_I(t)dW(t) -dM_1(t)
+\left\{ r_D(t)-r_f(t;\pi^f(t))\right\}dt\right]  \nonumber \\
&+ \pi^C(t) P_C(t-)\left[ \sigma_C(t)dW(t) -dM_2(t) 
+\left\{ r_D(t)-r_f(t;\pi^f(t))\right\}dt\right]  \nonumber \\
&+\left\{ Y(t)-\alpha \hat{V}(t)\right\} r_f(t;\pi^f(t))dt
+\alpha \hat{V}(t) r_{col}(t;\pi^{col}(t)) dt,
\end{align}
where we denote 
\[
 r_j(t;p):= r_j^{-}(t) 1_{\{ p<0\}} 
+r_j^+(t) 1_{\{ p>0\}},
\quad j\in\{ f,r,col\}.
\]
\begin{rem}
Suppose that $r_D\equiv r_f^\pm\equiv r_r^\pm \equiv r_{col}^\pm$.
Then (27) becomes 
\begin{align*}
 dY(t) =& \pi(t)^\top \text{\rm diag}\left(S(t)\right)\sigma(t)dW(t) 
+ \pi^I(t) P_I(t-)\left\{ \sigma_I(t)dW(t) -dM_1(t)\right\}  \\
&+ \pi^C(t) P_C(t-)\left\{ \sigma_C(t)dW(t) -dM_2(t) \right\}
+r_D(t)Y(t)dt,  
\end{align*}
which is solved as
\begin{multline}
Y^{(y,\Pi)}(t)= B_D(t)
\Biggl[
y+ \int_0^t
B_D(s)^{-1}\pi(s)^\top \text{\rm diag}\left(S(s)\right)\sigma(s)dW(s) \\
+\int_0^t B_D(s)^{-1}\pi^I(s)P_I(s-)\left\{\sigma_I(s)dW(s) -dM_1(s)\right\} \\
+\int_0^t B_D(s)^{-1}\pi^C(s)P_C(s-)\left\{\sigma_C(s)dW(s) -dM_2(s)\right\} 
\Biggr].
\end{multline}
That is, the discounted value process $Y/B_D$ is a local martingale, 
which is a standard result shared in a classical framework 
with ``one risk-free rate world.''
\end{rem}
For the derivative security given in Definition~2, 
we call the portfolio strategy 
$(\hat{y},\hat{\Pi})\in {\mathbb R} \times ({\mathcal P}_{{\mathbb G},T})^{n+2}$ 
that satisfies
\begin{equation}
 Y^{(\hat{y},\hat{\Pi})}_{\tau_1\wedge\tau_2\wedge T}=H
\end{equation}
the replicating portfolio strategy for the hedger. 

Furthermore, for pricing purposes, we next consider 
a dynamic portfolio strategy 
$\bigl( -\tilde{\pi},-\tilde{\pi}^I, -\tilde{\pi}^C, 
\tilde{\pi}^f,\tilde{\pi}^r, \tilde{\pi}^{col}\bigr)$ and the 
associated value process $\tilde{Y}$ of the buyer (counterparty).
We define 
\[
\begin{split}
 d\tilde{Y}(t) =& -\tilde{\pi}(t)^\top dS(t)
-\tilde{\pi}^I(t) dP_I(t) - \tilde{\pi}^C(t) dP_C(t) \\
&+\tilde{\pi}^f(t) dB_f(t) + \tilde{\pi}^r(t) dB_r(t) 
+\tilde{\pi}^{col}(t) dB_{col}(t), \\
 \tilde{Y}(0)=&-\tilde{y} 
\end{split}
\]
subject to 
\begin{align}
&\tilde{Y}(t) = -\tilde{\pi}(t)^\top S(t)
-\tilde{\pi}^I(t) P_I(t) - \tilde{\pi}^C(t) P_C(t) \nonumber \\
&\qquad\qquad+\tilde{\pi}^f(t) B_f(t) +\tilde{\pi}^r(t) B_r(t) 
+\tilde{\pi}^{col}(t) B_{col}(t), \\
&-\tilde{\pi}(t)^\top S(t) +\tilde{\pi}^r(t) B_r(t) =0, \\
&\tilde{\pi}^{col}(t) B_{col}(t)+\alpha \hat{V}(t)=0, 
\end{align}
where $\tilde{\pi}\in \left( {\mathcal P}_{{\mathbb G},T}\right)^n$
and $\tilde{\pi}^i\in {\mathcal P}_{{\mathbb G},T}$
for $i\in\{ I,C,f,r,col\}$.
Here, as we see in (32), 
the collateral value at time $t$ is regarded as $-\alpha\hat{V}(t)$,
the opposite value of that for the writer (hedger).
Because we see that
\begin{align*}
\tilde{\pi}^r(t)  =&B_r(t)^{-1}\tilde{\pi}(t)^\top S(t), \\
\tilde{\pi}^{col}(t) =&-B_{col}(t)^{-1}\alpha \hat{V}(t), \\
\tilde{\pi}^f(t)  =&
B_f(t)^{-1}\left\{
\tilde{Y}(t-) +\tilde{\pi}^I(t) P_I(t-) 
+\tilde{\pi}^C(t) P_C(t-)+\alpha \hat{V}(t)\right\}
\end{align*}
from (30)--(32), 
we regard $\bigl(-\tilde{y},-\tilde{\Pi}\bigr)
\in {\mathbb R}\times \left( {\mathcal P}_{{\mathbb G},T}\right)^{n+2}$
with
$\tilde{\Pi}:=\left(\tilde{\pi},\tilde{\pi}^I,\tilde{\pi}^C\right)$
as the portfolio strategy, and we rewrite the dynamics of 
$\tilde{Y}:\equiv \tilde{Y}^{(-\tilde{y},-\tilde{\Pi})}$
as
\begin{align}
 d\tilde{Y}(t) =& -\tilde{\pi}(t)^\top \text{diag}\left(S(t)\right)
\left[ \sigma(t)dW(t) + 
\left\{ r_D(t)-r_r(t;\pi^r(t))\right\}{\bf 1}dt\right]
\nonumber \\
&-\tilde{\pi}^I(t) P_I(t-)\left[ \sigma_I(t)dW(t) -dM_1(t)
+\left\{ r_D(t)-r_f(t;\pi^f(t))\right\}dt\right]  \nonumber \\
&-\tilde{\pi}^C(t) P_C(t-)\left[ \sigma_C(t)dW(t) -dM_2(t) 
+\left\{ r_D(t)-r_f(t;\pi^f(t))\right\}dt\right]  \nonumber \\
&+\left\{ \tilde{Y}(t)+\alpha \hat{V}(t)\right\} r_f(t;\pi^f(t))dt
-\alpha \hat{V}(t) r_{col}(t;\pi^{col}(t)) dt.
\end{align}
\begin{rem}
We have assumed that 
the funding rate $r_{f,I}^\pm$ for the investor (writer) 
and the funding rate $r_{f,C}^\pm$ for the counterparty (buyer)
are identical, i.e., $r_f^{\pm}\equiv r_{f,I}^\pm \equiv r_{f,C}^\pm$, 
which is a restrictive situation. 
However, without such an assumption, 
it looks difficult and complicated to derive 
an explicit sufficient condition to ensure the no-arbitrage property 
(see Theorem~4 and its proof).
\end{rem}
\begin{rem}
Suppose that $r_D\equiv r_f^\pm\equiv r_r^\pm \equiv r_{col}^\pm$.
Using a similar calculation to that in Remark~9, we solve (33) to see that 
$\tilde{Y}^{(-y',-\tilde{\Pi})}\equiv -Y^{(y',\tilde{\Pi})}$, 
where the right-hand side $Y^{(y',\tilde{\Pi})}$ is given by (28)
by letting $y:=y'$ and $\Pi:\equiv \tilde{\Pi}$.
\end{rem}
If the portfolio strategy 
$(-\tilde{y},-\tilde{\Pi})
\in {\mathbb R} \times ({\mathcal P}_{{\mathbb G},T})^{n+2}$ 
satisfies
\begin{equation}
 \tilde{Y}^{(-\tilde{y},-\tilde{\Pi})}_{\tau_1\wedge\tau_2\wedge T}=-H
\end{equation}
for the derivative security given in Definition~2, 
then we call it
the replicating portfolio strategy for the buyer.

\subsection{Deriving BSDE}

The replicating portfolio
$(\hat{y},\hat{\Pi})$ that satisfies (29) 
is represented using the solution to a BSDE.
Let 
\begin{align*}
Y^+:\equiv& Y^{(\hat{y},\hat{\Pi})}, \\
U_1^+(t):=& -\pi^I(t) P_I({t-}), \\
U_2^+(t):=& -\pi^C(t) P_C({t-}), \\
Z^+(t):=&
\sigma(t)^\top \text{diag}\left(S(t)\right)\pi(t) 
-U_1^+(t) \sigma_I(t)^\top 
-U_2^+(t) \sigma_C(t)^\top.
\end{align*}
Recalling
\[
\pi^f(t) B_f(t) 
=Y^+(t) +U_1^+(t)+U_2^+(t) -\alpha \hat{V}(t),
\]
we see that 
$\pi^f(t) \ge 0$ (resp.\ $\le 0$) is equivalent to
\[
 Y^+(t) +U_1^+(t)+U_2^+(t) -\alpha \hat{V}(t)\ge 0, \ 
\text{(resp.\ $\le 0$).}
\]
Also, recalling 
\begin{align*}
-\pi^r(t) B_r(t)
=&\pi(t)^\top \text{diag}(S(t)){\bf 1} \\
=&\left\{Z^+(t)^\top + U_1^+(t) \sigma_I(t)+U_2^+(t) \sigma_C(t)\right\}
\sigma(t)^{-1} {\bf 1},
\end{align*}
we see that
$\pi^r(t) \ge 0$ (resp.\ $\le 0$) is equivalent to
\[
\left\{Z^+(t)^\top + U_1^+(t) \sigma_I(t)+U_2^+(t) \sigma_C(t)\right\}
\sigma(t)^{-1} {\bf 1}\le 0 
\ \text{(resp.\ $\ge 0$)}.
\]
Using these relations, 
we then rewrite (27) as
\begin{multline*}
 dY^+(t) = -f^+\left(t,Y^+(t),Z^+(t),U_1^+(t),U_2^+(t);\hat{V}(t)\right)dt \\
+Z^+(t)^\top dW(t)
+U_1^+(t) dM_1(t)
+U_2^+(t) dM_2(t),
\end{multline*}
where 
\begin{align}
f^+\left(t,y,z,u_1,u_2;\hat{v}\right) 
:=&f^0\left( t,y,z,u_1,u_2\right) 
+\alpha\left\{r_f^0(t)\hat{v}
-r_{col}^+(t) \hat{v}^++r_{col}^{-}(t) \hat{v}^-\right\} 
\nonumber \\
&+\epsilon_f(t) \left| y+ u_1+u_2-\alpha\hat{v}\right| 
\nonumber \\
&+\epsilon_r(t) \left| 
\left\{ z^\top +u_1 \sigma_I(t) +u_2\sigma_C(t)\right\}
\sigma(t)^{-1} {\bf 1}\right|, 
\end{align}
with 
\begin{align}
&f^0\bigl( t,y,z,u_1,u_2\bigr) 
:= -r_f^0(t)y
+\left\{ r_r^0(t)-r_D(t)\right\}
z^\top \sigma(t)^{-1}{\bf 1}
\nonumber \\
&+\left[ -\left\{ r_f^0(t)-r_D(t)\right\}
+\left\{ r_r^0(t)-r_D(t)\right\}\sigma_I(t)\sigma(t)^{-1} {\bf 1}
\right]u_1 \nonumber \\
&+\left[ -\left\{ r_f^0(t)-r_D(t)\right\}
+\left\{ r_r^0(t)-r_D(t)\right\}\sigma_C(t)\sigma(t)^{-1} {\bf 1}
\right]u_2. 
\end{align}
So, we consider the BSDE
on the filtered probability space
$(\Omega,{\mathcal F},{\mathbb P},{\mathbb G})$, namely 
\begin{equation}
\begin{split}
-dY^+(t) =&f^+\left(t,Y^+(t),Z^+(t),U_1^+(t),U_2^+(t);\hat{V}(t)\right)dt \\
&-Z^+(t)^\top dW(t) -U_1^+(t) dM_1(t) -U_2^+(t) dM_2(t) \\
&\qquad\text{for}\quad 0\le t\le \tau_1\wedge \tau_2 \wedge T, \\
Y^+\left(\tau_1 \wedge \tau_2 \wedge T\right)=&H.
\end{split}
\end{equation}
Using the solution to (37), 
the replicating portfolio 
$\left(\hat{y},\hat{\Pi}\right)$ that satisfies (29) 
is constructed as
\begin{align*}
\hat{y}:=&Y^+(0), \\
\hat{\pi}(t):=& 
\text{diag}(S_t)^{-1} 
\left( \sigma(t)^\top\right)^{-1}
\left\{
Z^+(t)+U_1^+(t)\sigma_I^\top(t)+U_2^+(t)\sigma_C^\top(t)
\right\}, \\
\hat{\pi}^I(t):=&-{P_I(t-)^{-1}}{U_1^+(t)}, \\
\hat{\pi}^C(t):=&-{P_C(t-)^{-1}}{U_2^+(t)}
\end{align*}
for $0\le t\le \tau_1\wedge \tau_2 \wedge T$.
Similarly, 
the replicating portfolio $(-\tilde{y},-\tilde{\Pi})$ 
that satisfies (34) 
can be represented using the solution to a BSDE.
Let 
\begin{align*}
Y^-:\equiv& -\tilde{Y}^{(-\tilde{y},-\tilde{\Pi})}, \\
U_1^-(t):=& -\tilde{\pi}^I(t) P_I({t-}), \\
U_2^-(t):=& -\tilde{\pi}^C(t) P_C({t-}), \\
Z^-(t):=&
\sigma(t)^\top \text{diag}\left(S(t)\right)\tilde{\pi}(t) 
-U_1^-(t) \sigma_I(t)^\top 
-U_2^-(t) \sigma_C(t)^\top.
\end{align*}
Recalling
\[
-\tilde{\pi}^f(t) B_f(t) 
=\tilde{Y}^-(t) +U_1^-(t)+U_2^-(t) -\alpha \hat{V}(t),
\]
we see that 
$\pi^f(t) \ge 0$ (resp.\ $\le 0$) is equivalent to
\[
 Y^-(t) +U_1^-(t)+U_2^-(t) -\alpha \hat{V}(t)\le 0, \ 
\text{(resp.\ $\ge 0$).}
\]
Also, recalling 
\begin{align*}
\tilde{\pi}^r(t) B_r(t)
=&\tilde{\pi}(t)^\top \text{diag}(S(t)){\bf 1} \\
=&\left\{Z^-(t)^\top + U_1^-(t) \sigma_I(t)+U_2^-(t) \sigma_C(t)\right\}
\sigma(t)^{-1} {\bf 1},
\end{align*}
we see that
$\pi^r(t) \ge 0$ (resp.\ $\le 0$) is equivalent to
\[
\left\{Z^-(t)^\top + U_1^-(t) \sigma_I(t)+U_2^-(t) \sigma_C(t)\right\}
\sigma(t)^{-1} {\bf 1}\ge 0 
\ \text{(resp.\ $\le 0$)}.
\]
Using these relations, 
we then rewrite (33) as
\begin{multline*}
 dY^-(t) = -f^-\left(t,Y^-(t),Z^-(t),U_1^-(t),U_2^-(t);\hat{V}(t)\right)dt \\
+Z^-(t)^\top dW(t)
+U_1^-(t) dM_1(t)
+U_2^-(t) dM_2(t),
\end{multline*}
where 
\begin{align}
f^-\left(t,y,z,u_1,u_2;\hat{v}\right) 
:=&-f^+\left( t,-y,-z,-u_1,-u_2;-\hat{v}\right) 
\nonumber \\
=&f^0\left( t,y,z,u_1,u_2\right) 
+\alpha\left\{r_f^0(t)\hat{v}
+r_{col}^+(t) \hat{v}^--r_{col}^{-}(t) \hat{v}^+\right\} 
\nonumber \\
&-\epsilon_f(t) \left| y+ u_1+u_2-\alpha\hat{v}\right| 
\nonumber \\
&-\epsilon_r(t) \left| 
\left\{ z^\top +u_1 \sigma_I(t) +u_2\sigma_C(t)\right\}
\sigma(t)^{-1} {\bf 1}\right|.
\end{align}
So, we consider the BSDE
on the filtered probability space
$(\Omega,{\mathcal F},{\mathbb P},{\mathbb G})$
\begin{equation}
\begin{split}
-dY^-(t) =&f^-\left(t,Y^-(t),Z^-(t),U_1^-(t),U_2^-(t);\hat{V}(t)\right)dt \\
&-Z^-(t)^\top dW(t) -U_1^-(t) dM_1(t) -U_2^-(t) dM_2(t) \\
&\qquad\text{for}\quad 0\le t\le \tau_1\wedge \tau_2 \wedge T, \\
Y^-\left(\tau_1 \wedge \tau_2 \wedge T\right)=&H.
\end{split}
\end{equation}
The replicating portfolio 
$\left(-\tilde{y},-\tilde{\Pi}\right)$ that satisfies (34) 
is now constructed as
\begin{align*}
\tilde{y}:=&Y^-(0), \\
\tilde{\pi}(t):=& 
\text{diag}(S_t)^{-1} 
\left( \sigma(t)^\top\right)^{-1}
\left\{
Z^-(t)+U_1^-(t)\sigma_I^\top(t)+U_2^-(t)\sigma_C^\top(t)
\right\}, \\
\tilde{\pi}^I(t):=&-{P_I(t-)^{-1}}{U_1^-(t)}, \\
\tilde{\pi}^C(t):=&-{P_C(t-)^{-1}}{U_2^-(t)}
\end{align*}
for $0\le t\le \tau_1\wedge \tau_2 \wedge T$, 
using the solution to (39).

\begin{rem}
BSDEs (37) and (39) with (15) and (16) can be seen as the system of BSDEs 
\begin{equation}
\begin{split}
-dY^{\pm}(t) =&f^{\pm}\left(t,Y^{\pm}(t),
Z^{\pm}(t),U_1^{\pm}(t),U_2^{\pm}(t);\hat{V}(t)\right)dt \\
&-Z^{\pm}(t)^\top dW(t) -U_1^{\pm}(t) dM_1(t) -U_2^{\pm}(t) dM_2(t), \\
&\qquad\text{for}\quad 0\le t\le \tau_1\wedge \tau_2 \wedge T, \\
Y^{\pm}\left(\tau_1 \wedge \tau_2 \wedge T\right)=&H, \\
-d\hat{V}(t)=& -r_D(t) \hat{V}(t)dt -\Delta(t)^\top dW(t)
\quad\text{for}\quad 0\le t\le T, \\
\hat{V}(T)=&\xi_T,
\end{split}
\end{equation}
in which $\left(Y^{\pm},Z^{\pm},
U_1^{\pm},U_2^{\pm},\hat{V},\Delta \right)$ are solutions.
\end{rem}

\subsection{Hedging Problem}

To study the hedging problem via BSDEs (37) and (39) with (15) and (16), 
it is natural to employ the following 
space of admissible hedging strategies
\[
{\mathscr A}_{\beta,T}:=\left\{ 
\left( \pi, \pi^I, \pi^C\right)\in 
\left(
{\mathcal P}_{{\mathbb G},T}\right)^{d+2}
\Bigm|
\left(
\sigma^\top \text{diag}(S) \pi, 
\pi^I P_{I}^-,
\pi^C P_{C}^-
\right)
\in {\mathbb H}_{\beta,T}^{2,n+2}
\right\},
\]
where $\beta>0$ is a fixed (sufficiently large) constant 
and we denote $P_i^-(t):=P_i(t-)$ for $t>0$ and $P_i^-(0):=P_i(0)$.
We then formulate the minimal superhedging price 
(i.e., the maximal price for the writer)
and the maximal subhedging price 
(i.e., the minimal price for the buyer)
as follows.
\begin{defi}
For the derivative security given in Definition~2, 
\[
\bar{p}:= \inf\left\{ y\in {\mathbb R}
\bigm| -H+Y^{(y,\Pi)}(\tau_1\wedge\tau_2 \wedge T)\ge 0
\ \text{for some $(y,\Pi)\in{\mathbb R}\times {\mathscr A}_{\beta,T}$}\right\}
\]
is called the minimal superhedging price, 
which is the maximal price of the writer (seller), and
\[
\underline{p}:= \sup\left\{ y\in {\mathbb R}
\bigm| H+\tilde{Y}^{(-y,-\Pi)}(\tau_1\wedge\tau_2 \wedge T)\ge 0
\ \text{for some $(y,\Pi)\in{\mathbb R}\times {\mathscr A}_{\beta,T}$}\right\}
\]
is called the maximal subhedging price, 
which is the minimal price of the buyer.
If there exists $\bar{\Pi}\in {\mathscr A}_{\beta,T}$ such that
\[
 -H+ Y^{\left(\bar{p},\bar{\Pi} \right)}(\tau_1\wedge\tau_2 \wedge T)\ge 0,
\]
then 
the pair $\left(\bar{p},\bar{\Pi} \right)$
is called the minimal superhedging strategy, and 
if there exists $\underline{\Pi}\in {\mathscr A}_{\beta,T}$ such that
\[
 H+
\tilde{Y}^{\left(-\underline{p},-\underline{\Pi} \right)}
(\tau_1\wedge\tau_2 \wedge T)\ge 0,
\]
then 
the pair $\left(-\underline{p},-\underline{\Pi} \right)$
is called the maximal subhedging strategy.
\end{defi}
Associated with the hedging problem, we give the following definition.
\begin{defi}
Consider the derivative security given in Definition~2.
Suppose that
a writer sells the derivative security with price $p\in {\mathbb R}$
at time $0$. If it holds that
\[
 -H + Y^{(p,\Pi)}(\tau_1\wedge\tau_2 \wedge T)\ge 0
\quad\text{and}\quad
{\mathbb P}\left(
 -H + Y^{(p,\Pi)}(\tau_1\wedge\tau_2 \wedge T)
> 0
\right)
>0
\]
for some $\Pi\in {\mathscr A}_{\beta,T}$, 
then we say that an arbitrage opportunity for the writer occurs. 
Similarly, suppose that
a buyer purchases the derivative security with price $p\in {\mathbb R}$
at time $0$. If it holds that
\[
 H + \tilde{Y}^{(-p,-\Pi)}(\tau_1\wedge\tau_2 \wedge T)\ge 0
\quad\text{and}\quad
{\mathbb P}\left(
 H + \tilde{Y}^{(-p,-\Pi)}(\tau_1\wedge\tau_2 \wedge T)
> 0
\right)
>0
\]
for some $\Pi\in {\mathscr A}_{\beta,T}$, 
then we say that an arbitrage opportunity for the buyer occurs. 
Moreover, if the price $\hat{p}\in {\mathbb R}$ at time $0$ 
does not admit arbitrage opportunities for both writer and buyer, 
then $\hat{p}$ is called an arbitrage-free price.
\end{defi}

\begin{rem}
In our financial market model,  we assume implicitly that 
the probability measure $\mathbb P$ is an EMM. 
Hence, ${\mathbb P}\sim {\mathbb P}_0$, where 
${\mathbb P}_0$ is a real-world (physical) probability measure
given in the same measurable space $(\Omega,{\mathcal F})$.
Therefore, in Definition~3, the $\mathbb P$-a.s. statement can be replaced by 
the ${\mathbb P}_0$-a.s. statement.
Also, in Definition~4, $\mathbb P$ can be replaced by ${\mathbb P}_0$
to claim that 
${\mathbb P}_0\left( \cdots \right)>0$.
\end{rem}

\subsection{Markovian Model}

The following Markovian model is typical and popularly treated in practice.
Let the coefficients of the market model be described as
\begin{align*}
\sigma(t):=&\tilde{\sigma}\left( t,F(t)\right), \quad
r_D(t):=\tilde{r}_D\left( t,F(t)\right), \\
\sigma_i(t):=&\tilde{\sigma}_j \left( t,F(t)\right), \quad
i\in\{ I,C\}, \\
h_j(t):=&\tilde{h}_i\left( t,F(t)\right), \quad
j\in\{ 1,2\}, \\
r_k^{0}(t):=& \tilde{r}_{k}^{0}\left( t,F(t)\right), \quad
\epsilon_k(t):= \tilde{\epsilon}_{k}\left( t,F(t)\right), \quad
k\in\{ f,r\}, \\
\text{and}\quad
r_{col}^{\pm}(t):=& \tilde{r}_{col}^{\pm}\left( t,F(t)\right), 
\end{align*}
where
$\tilde{\sigma}:[0,T]\times {\mathbb R}^m \to {\mathbb R}^{n\times n}$,
$\tilde{r}_D,\tilde{\sigma}_i,\tilde{h}_j, \tilde{r}_{k}^0,
\tilde{\epsilon}_{k}, \tilde{r}_{col}^\pm:
[0,T]\times {\mathbb R}^m \to {\mathbb R}$,
and $(F(t))_{t\in [0,T]}$ is called the stochastic factor process, 
which can be interpreted as a model of economic factors and 
affects the market model
through the coefficients 
$\sigma,\sigma_i$ ($i\in\{I,C\}$), 
$h_j$ ($j=1,2$), 
$r_k^0, \epsilon_k^0$ ($k\in\{ f,r\}$), 
and $r_{col}^\pm$. 
It is given by the solution to the SDE 
\[
 dF(t)= \mu_F\left( t,F(t)\right)dt
+\sigma_F\left( t,F(t)\right)dW(t),
\quad 
F(0)\in {\mathbb R}^m
\]
on $(\Omega,{\mathcal F},{\mathbb P}, {\mathbb F})$, 
where 
$\mu_F:[0,T]\times {\mathbb R}^m \to {\mathbb R}^m$ and
$\sigma_F:[0,T]\times {\mathbb R}^m \to {\mathbb R}^{m\times n}$.
Let
\[
 X^\top:\equiv \left( X_1^\top, X_2^\top\right)
:\equiv \left( S^\top, F^\top\right)
\]
and define, for $x:=(x_1,x_2)\in {\mathbb R}^n \times {\mathbb R}^m$, 
\[
b(t,x):=
\begin{pmatrix}
\text{diag}(x_1) r_D(t,x_2) \\
\mu_F(t,x_2)
\end{pmatrix},
\quad
a(t,x):=
\begin{pmatrix}
\text{diag}(x_1)\sigma(t,x_2) \\
\sigma_F(t,x_2)
\end{pmatrix}.
\]
Then, the SDE for $X$ is written as (8) with $d=n+m$. 
Furthermore, we set 
\[
 \xi_T:=\Xi\left( X(T)\right)
\quad\text{and}\quad
 \phi_i(t):=\varphi_i( \hat{V}(t))
\quad \text{for $i\in\{1,2\}$,}
\]
where $\Xi:{\mathbb R}^{n+m} \to {\mathbb R}$
and $\varphi_i:{\mathbb R} \to {\mathbb R}$.
In this situation, we can apply Theorem~3 to represent 
the solution to BSDEs (40), 
using the solutions to the associated PDEs 
(see Proposition~2 in Section~4).

\section{Results}

Throughout this section, we always assume that
$\sigma_i$ ($i\in\{I,C\}$), $\sigma, \sigma^{-1}$, 
$r_D$, $r_{j}^\pm$ ($j\in\{ f,r,col\}$), and $h_k$ ($k=1,2$)
are bounded. 
Applying the results in Section~2 and a comparison theorem for BSDEs, 
the following claims are straightforward to see.
\begin{prop}
For any sufficiently large $\beta>0$, 
there exist unique solutions
$\left( Y^{\pm}, Z^{\pm}, U_1^{\pm}, U_2^{\pm}\right)
\in {\mathbb S}^2_{\beta,T} \times {\mathbb H}^{2,n+2}_{\beta,T}$
to BSDEs (37) and (39) with (15) and (16).
Moreover, the solutions have the representations that
\begin{equation}
\begin{split}
Y^\pm(t)=&
\bar{Y}^\pm(t)1_{\{ 0\le t<\tau_1\wedge\tau_2\wedge T \}} \\ 
+ \Bigl\{
\phi_1&({\tau_1}) 1_{\{ \tau_1<\tau_2\wedge T\}}
+\phi_2({\tau_2}) 1_{\{ \tau_2<\tau_1\wedge T\}}
+\xi_T 1_{\{ T<\tau_1 \wedge \tau_2\}}
\Bigr\}
1_{\{ t=\tau_1 \wedge \tau_2\wedge T\}}, \\
Z^\pm(t)=& \bar{Z}^\pm(t), \\
U_i^\pm(t)=& \phi_i(t)-\bar{Y}^\pm(t), \quad i=1,2.
\end{split}
\end{equation}
Here,  
$\left( \bar{Y}^\pm,\bar{Z}^\pm\right)
\in {\mathbb S}^2_{\beta,T} \times {\mathbb H}^{2,n}_{\beta,T}$ 
are the solutions to BSDEs
on $(\Omega,{\mathcal F},{\mathbb P}, {\mathbb F})$, namely 
\begin{equation}
\begin{split}
-d\bar{Y}^\pm(t) =&\bar{f}^\pm\left(t,\bar{Y}^\pm(t),\bar{Z}^\pm(t)
;\hat{V}(t), \phi_1(t),\phi_2(t)\right)dt 
-\bar{Z}^\pm(t)^\top dW(t) \\
&\qquad\text{for}\quad 0\le t\le T, \\
\bar{Y}^\pm(T)=&\xi_T,  \\
-d\hat{V}(t)=& -r_D(t) \hat{V}(t)dt -\Delta(t)^\top dW(t)
\quad\text{for}\quad 0\le t\le T, \\
\hat{V}(T)=&\xi_T,
\end{split}
\end{equation}
where we define
\begin{align}
\bar{f}^\pm\left( t,y,z; \hat{v},p_1,p_2\right) 
:=f^\pm\left( t,y,z, p_1-y,p_2-y; \hat{v}\right)
+(p_1-y)h_1(t)+(p_2-y)h_2(t).
\end{align}
In addition to Condition~(19), assume that
\begin{equation}
 r_{col}^-\ge r_{col}^+. 
\end{equation}
Then, it always holds that 
\begin{equation}
 Y^-\le Y^+
\quad\text{and}\quad
\bar{Y}^-\le \bar{Y}^+.
\end{equation}
\end{prop}

\begin{proof}[Sketch]
Using (19) and (44), we see that
\begin{align*}
&\bar{f}^+\left( t,y,z; \hat{v},p_1,p_2\right) 
-\bar{f}^-\left( t,y,z; \hat{v},p_1,p_2\right) \\
=&\alpha\left\{ r_{col}^-(t)-r_{col}^+(t)\right\} |\hat{v}|
+ 2\epsilon_f(t) \left| y+ (p_1-y)+ (p_2-y)-\alpha \hat{v}\right| \\
&+ 2\epsilon_r(t) \left|
\left\{  z^\top +(p_1-y)\sigma_I(t) +(p_2-y)\sigma_C(t)\right\}
\sigma(t)^{-1}{\bf 1}\right|\ge 0.
\end{align*}
Hence, (45) follows from a comparison theorem of BSDEs. 
Other assertions follow from the results in Section~2.
\end{proof}
Next, consider the Markovian model given in Subsection~3.6. 
Then, corresponding to (42), we have the Markovian system of BSDEs 
(decoupled FBSDEs) 
\begin{equation}
\begin{split}
dX(t)=&b(t,X(t))dt + a(t,X(t))dW(t),
\quad X(0)\in {\mathbb R}^{n+m}, \\
-d\bar{Y}^{\pm}(t) =&\bar{g}^{\pm}\left(t,X_2(t), \bar{Y}^{\pm}(t),
\bar{Z}^{\pm}(t); \hat{V}(t), \varphi_1(\hat{V}(t)), 
\varphi_2(\hat{V}(t))\right)dt \\
&-\bar{Z}^{\pm}(t)^\top dW(t), \\
\bar{Y}^{\pm}(T)=&\Xi\left( X(T)\right), \\
-d\hat{V}(t)=& -\tilde{r}_D(t,X_2(t)) \hat{V}(t)dt -\Delta(t)^\top
 dW(t), \\
\hat{V}(T)=&\Xi\left( X(T)\right).
\end{split}
\end{equation}
Here, the relation 
\[
\bar{g}^{\pm}(t,X_2(t,\omega),y,z;\hat{v},p_1,p_2)
=\bar{f}^{\pm}(t,\omega,y,z;\hat{v},p_1,p_2)
\]
holds, and the functions 
$\bar{g}^\pm:[0,T]\times {\mathbb R}^m \times 
{\mathbb R}\times {\mathbb R}^n \times {\mathbb R}^3 
\to {\mathbb R}$ are written as
\begin{align*}
&\bar{g}^\pm(t,x_2,y,z;\hat{v},p_1,p_2)
:=\bar{g}^0 (t,x_2,y,z;p_1,p_2) \\
&+\alpha\left\{ 
\tilde{r}_f^0(t,x_2) \hat{v} 
\mp \tilde{r}_{col}^+(t,x_2)\hat{v}^\pm 
\pm \tilde{r}_{col}^-(t,x_2) \hat{v}^\mp
\right\} \\
&\pm\tilde{\epsilon}_f(t,x_2)\left|y+(p_1-y)+(p_2-y)-\alpha\hat{v}\right| \\
&\pm\tilde{\epsilon}_r(t,x_2)
\left| 
\left\{ 
z^\top +(p_1-y)\tilde{\sigma}_I(t,x_2)
+(p_2-y)\tilde{\sigma}_C(t,x_2)
\right\}
\tilde{\sigma}(t,x_2)^{-1} {\bf 1}
\right|
\end{align*}
with
\begin{align*}
&\bar{g}^0 (t,x_2,y,z;p_1,p_2)  
:=z^\top \left\{ (\tilde{r}_r^0-\tilde{r}_D)
\tilde{\sigma}^{-1} {\bf 1} \right\}(t,x_2) \\
&-\left\{ (2\tilde{r}_D-\tilde{r}_f^0+\tilde{h}_1+\tilde{h}_2)
+(\tilde{r}_r^0-\tilde{r}_D)
(\tilde{\sigma}_I+\tilde{\sigma}_C)
\tilde{\sigma}^{-1}{\bf 1}\right\}(t,x_2) y \\
&+\left\{\tilde{h}_1-(\tilde{r}_f^0-\tilde{r}_D)
+(\tilde{r}_r^0-\tilde{r}_D)\tilde{\sigma}_I
\tilde{\sigma}^{-1}{\bf 1}\right\}(t,x_2)p_1 \\
&+\left\{ \tilde{h}_2-
(\tilde{r}_f^0-\tilde{r}_D)+(\tilde{r}_r^0-\tilde{r}_D)
\tilde{\sigma}_C\tilde{\sigma}^{-1}{\bf 1}\right\}(t,x_2)p_2.
\end{align*}
Utilizing Theorem~3, we obtain the following.
\begin{prop}
Denote $d:=n+m$ and 
consider the system of second-order parabolic semilinear PDEs 
\begin{equation}
\begin{split}
-\partial_t V=&
\left\{{\mathcal L}_t-\tilde{r}_D(t,x_2)
\right\}V, \quad
(t,x)\in [0,T)\times {\mathbb R}^d, \\
V(T,x)=&\Xi(x), \\
-\partial_t U^\pm=&{\mathcal L}_t U^\pm 
+ \bar{g}^\pm \bigl( t,x_2, U, a^\top \nabla U^\pm;
V, \varphi_1(V), \varphi_2(V)\bigr), \\
& (t,x)\in [0,T)\times {\mathbb R}^d, \\
U^\pm(T,x)=&\Xi(x),
\end{split}
\end{equation}
where ${\mathcal L}_t(\cdot)$ 
is the infinitesimal generator for $X$ given by (10).
Suppose that there exists a unique classical solution 
$\left(V,U^\pm\right)\in 
\left(C^{1,2}([0,T]\times {\mathbb R}^d)\right)^2$ to (47). 
Then the solution to BSDE (46) is represented as
\[
 \bar{Y}^\pm(t)=U^\pm\left( t,X(t)\right),
\quad
 \bar{Z}^\pm(t)=\left( a \nabla U^\pm\right)\left( t,X(t)\right),
\quad
t\in [0,T].
\]

\end{prop}
\subsection{Results on Arbitrage}

\begin{thm}
In addition to Conditions~(19) and (44), assume the following:
\begin{equation}
\begin{split}
h_1\ge& r_f^- - r_D 
-\left( r_r^+ -r_D\right)(\sigma_I\sigma^{-1}{\bf 1})^+ 
+\left( r_r^- -r_D\right)(\sigma_I\sigma^{-1}{\bf 1})^-, \\
h_2\ge& r_f^- - r_D 
-\left( r_r^+ -r_D\right)(\sigma_C\sigma^{-1}{\bf 1})^+ 
+\left( r_r^- -r_D\right)(\sigma_C\sigma^{-1}{\bf 1})^-,
\end{split}
\end{equation}
and 
\begin{equation}
r_f^+ \ge r_{col}^-.
\end{equation}
Then it holds that 
$\underline{p}=Y^-(0)\le Y^+(0)=\bar{p}$.
Hence, 
for the derivative security given in Definition~2, 
any price $p\in \left[ Y^-(0),Y^+(0)\right]$ at time $0$ is arbitrage-free.
\end{thm}

\begin{rem}
The conditions imposed in Theorem~4
to ensure the arbitrage-free property look to be rather strong: 
violating (44), (48), or (49) seems to be realizable 
in real situations. 
Relaxing the arbitrage-free condition
by admitting ``certain'' arbitrage opportunities 
might be an interesting research direction 
for this bilateral hedging scheme with collateralizations. 
We refer to Thoednithi (2015)
and Nie and Rutkowski (2018) as related studies.
\end{rem}
\begin{proof}[Sketch]
Using (35), (36), (38), and (43), we see that
\begin{align*}
&\bar{f}^\pm\left( t,y,z; \hat{v},p_1,p_2\right) \\
=&z^\top \left\{ (r_r^0-r_D)\sigma^{-1} {\bf 1} \right\}(t) \\
&-\left\{ (2r_D-r_f^0+h_1+h_2)
+(r_r^0-r_D)(\sigma_I+\sigma_C)\sigma^{-1}{\bf 1}\right\}(t) y \\
&+\left\{h_1-(r_f^0-r_D)+(r_r^0-r_D)\sigma_I\sigma^{-1}{\bf
 1}\right\}(t)p_1 \\
&+\left\{h_2-(r_f^0-r_D)+(r_r^0-r_D)\sigma_C\sigma^{-1}{\bf
 1}\right\}(t)p_2 \\
&+\alpha\left\{ 
r_f^0(t) \hat{v} \mp r_{col}^+(t)\hat{v}^\pm \pm r_{col}^-(t) \hat{v}^\mp
\right\} \\
&\pm \epsilon_f(t) \left| y+ (p_1-y)+ (p_2-y)-\alpha \hat{v}\right| \\
&\pm \epsilon_r(t) \left|
\left\{  z^\top +(p_1-y)\sigma_I(t) +(p_2-y)\sigma_C(t)\right\}
\sigma(t)^{-1}{\bf 1}\right|. 
\end{align*}
So, for $\delta_0,\delta_1,\delta_2\ge 0$, we see that
\begin{align}
&\bar{f}^+\left( \cdot,y,z; \hat{v}+\delta_0,p_1+\delta_1,p_2+\delta_2\right) 
-\bar{f}^+\left( \cdot,y,z; \hat{v},p_1,p_2\right) 
\nonumber \\
=&\left\{h_1-(r_f^0-r_D)+(r_r^0-r_D)\sigma_I\sigma^{-1}{\bf
 1}\right\}\delta_1 
\nonumber \\
&+\left\{h_2-(r_f^0-r_D)+(r_r^0-r_D)\sigma_C\sigma^{-1}{\bf
 1}\right\}\delta_2 
\nonumber \\
&+\alpha\left[
r_f^0 \delta_0 - r_{col}^+
\left\{ (\hat{v}+\delta_0)^+ -\hat{v}^+\right\} 
+ r_{col}^-
\left\{ (\hat{v}+\delta_0)^- -\hat{v}^-\right\} 
\right] 
\nonumber \\
&+ \epsilon_f
\left\{
\left| p_1+p_2 -\alpha\hat{v}-y +(\delta_1+\delta_2-\alpha\delta_0)\right| 
-\left| p_1+p_2 -\alpha\hat{v}-y \right| 
\right\}
\nonumber \\
&+ \epsilon_r
\Bigl[
\left|\left\{  z^\top +(p_1-y)\sigma_I +(p_2-y)\sigma_C\right\}
\sigma^{-1}{\bf 1}
+\left\{\delta_1 \sigma_I+\delta_2 \sigma_C\right\}\sigma^{-1}{\bf 1}
\right| 
\nonumber \\
&-\left|\left\{  z^\top +(p_1-y)\sigma_I +(p_2-y)\sigma_C\right\}
\sigma^{-1}{\bf 1}\right| 
\Bigr].
\end{align}
Using the inequality $|x+y|-|x|\ge -|y|$
and the relation
\[
r_{col}^+
\left\{ (\hat{v}+\delta_0)^+ -\hat{v}^+\right\} 
- r_{col}^-
\left\{ (\hat{v}+\delta_0)^- -\hat{v}^-\right\} 
\le
\left(  r_{col}^{+} \vee r_{col}^-\right)\delta_0, 
\]
we see that
\begin{align}
\text{(50)}\ge &
\left\{h_1-(r_f^0-r_D)+(r_r^0-r_D)\sigma_I\sigma^{-1}{\bf 1}\right\}\delta_1 
\nonumber \\
&+\left\{h_2-(r_f^0-r_D)+(r_r^0-r_D)\sigma_C\sigma^{-1}{\bf
 1}\right\}\delta_2 
+\alpha\left(r_f^0 - r_{col}^-\right)\delta_0
\nonumber \\
&- \epsilon_f \left( \delta_1+\delta_2+\alpha \delta_0 \right)
- \epsilon_r 
\left\{
|\sigma_I\sigma^{-1}{\bf 1}| \delta_1
+|\sigma_C\sigma^{-1}{\bf 1}| \delta_2
\right\} 
\nonumber \\
=&
\left\{h_1-r_f^-+r_D+(r_r^0-r_D)\sigma_I\sigma^{-1}{\bf 1}
-\epsilon_r|\sigma_I\sigma^{-1}{\bf 1}|
\right\}\delta_1 
\nonumber \\
+&\left\{h_2-r_f^-+r_D+(r_r^0-r_D)\sigma_C\sigma^{-1}{\bf 1}
-\epsilon_r|\sigma_C\sigma^{-1}{\bf 1}|
\right\}\delta_2 
\nonumber \\
+&\alpha\left(r_f^+-r_{col}^-\right)\delta_0 \ge 0,
\end{align}
where we use (48) and (49).
Consider the system of BSDEs (42) and write the solution as
\[
 \bar{Y}^{\pm}\left( t;\xi_T,\phi_1,\phi_2\right),
\quad
 \bar{Z}^{\pm}\left( t;\xi_T,\phi_1,\phi_2\right)
\quad t\in [0,T]
\]
by emphasizing the parameters
$\left({\xi}_T,\phi_1,{\phi}_2\right)$.
Take other payoff parameters \\
$\left( \tilde{\xi}_T, \tilde{\phi}_1, \tilde{\phi}_2\right)$
such that
$\tilde{\xi}_T\ge \xi_T$, 
$\tilde{\phi}_1\ge {\phi}_1$, and
$\tilde{\phi}_2\ge {\phi}_2$.
Using the comparison theorem for BSDEs twice
(for $\hat{V}$ and $\bar{Y}^+$), 
and using relations (50) and (51), we deduce that
\[
 \bar{Y}^+\left(\tilde{\xi}_T,\tilde{\phi}_1,\tilde{\phi}_2 \right)
\ge \bar{Y}^+\left({\xi}_T,{\phi}_1,{\phi}_2 \right)
\]
and that
\[
{Y}^+\left(\tilde{\xi}_T,\tilde{\phi}_1,\tilde{\phi}_2 \right)
\ge {Y}^+\left({\xi}_T,{\phi}_1,{\phi}_2 \right).
\]
This implies the minimality 
of $Y^+(\xi_T,\phi_1,\phi_2)$ and the equality, 
\[
\bar{p}=Y^+(0;\xi_T,\phi_1,\phi_2). 
\]
The equality, 
\[
\underline{p}=Y^-(0;\xi_T,\phi_1,\phi_2),
\]
can be seen similarly.
\end{proof}

\begin{rem}
We have that for $k\ge 0$, 
\[
 Y^\pm\left( t; k \xi_T, k\phi_1, k\phi_2\right)
\equiv k  Y^\pm\left( t; \xi_T, \phi_1, \phi_2\right)
\quad\text{for \  $t\in [0,T]$.}
\]
This positive homogeneity is seen from those 
of the drivers of BSDEs (42), namely
\begin{align*}
\bar{f}^\pm\left(t,ky,kz;k\hat{v},k\phi_1,k\phi_2\right)
=&k\bar{f}^\pm\left(t,y,z;\hat{v},\phi_1,\phi_2\right), \\
-r_D(t) \left( k\hat{v}\right)=&k \left\{ -r_D(t) \hat{v}\right\}.
\end{align*}
See Jiang (2008) for the details.
\end{rem}

\subsection{Results on XVA}

In this subsection, we assume that
\begin{equation}
\epsilon_f \vee \epsilon_r \le \epsilon 
\end{equation}
with some (small) positive constant $\epsilon\ll 1$. 
Consider the system of BSDEs
\begin{equation}
\begin{split}
-dY^{0,\pm}(t) =&f^{0,\pm}\left(t,Y^{0,\pm}(t),
Z^{0,\pm}(t),U_1^{0,\pm}(t),U_2^{0,\pm}(t);\hat{V}(t)\right)dt \\
-&Z^{0,\pm}(t)^\top dW(t) -U_1^{0,\pm}(t) dM_1(t) -U_2^{0,\pm}(t) dM_2(t), \\
&\qquad\text{for}\quad 0\le t\le \tau_1\wedge \tau_2 \wedge T, \\
Y^{0,\pm}\left(\tau_1 \wedge \tau_2 \wedge T\right)=&H, \\
-d\hat{V}(t)=& -r_D(t) \hat{V}(t)dt -\Delta(t)^\top dW(t)
\quad\text{for}\quad 0\le t\le T, \\
\hat{V}(T)=&\xi_T
\end{split}
\end{equation}
on $(\Omega,{\mathcal F},{\mathbb P}, {\mathbb G})$, where
\begin{align*}
f^{0,\pm}\left(t,y,z,u_1,u_2;\hat{v}\right) 
:=f^{0}\left(t,y,z,u_1,u_2\right) 
+\alpha\left\{ 
r_f^0(t) \hat{v} \mp r_{col}^+(t)\hat{v}^\pm \pm r_{col}^-(t) \hat{v}^\mp
\right\}.
\end{align*}
Associated with (53), consider the reduced system of BSDEs
\begin{equation}
\begin{split}
-d\bar{Y}^{0,\pm}(t) =&\bar{f}^{0,\pm}\left(t,
\bar{Y}^{0,\pm}(t),
\bar{Z}^{0,\pm}(t); \hat{V}(t), \phi_1(t), \phi_2(t)\right)dt \\
&-\bar{Z}^{0,\pm}(t)^\top dW(t)
\quad\text{for}\quad 0\le t\le T, \\
\bar{Y}^{0,\pm}(T)=&\xi_T, \\
-d\hat{V}(t)=& -r_D(t) \hat{V}(t)dt -\Delta(t)^\top dW(t)
\quad\text{for}\quad 0\le t\le T, \\
\hat{V}(T)=&\xi_T
\end{split}
\end{equation}
on $(\Omega,{\mathcal F},{\mathbb P}, {\mathbb F})$, where
\begin{align*}
\bar{f}^{0,\pm}\left(t,y,z;\hat{v},p_1,p_2\right) 
:=f^{0,\pm}\left(t,y,z,p_1-y,p_2-y;\hat{v}\right) 
+(p_1-y)h_1(t)+(p_2-y)h_2(t).
\end{align*}
We obtain the following.
\begin{thm}
Assume Conditions~(19) and (44).
For 
$(\bar{Y}^\pm, \bar{Z}^{\pm})$, 
$(\bar{Y}^{0,\pm}, \bar{Z}^{0,\pm})$, 
which are solutions to BSDEs (42) and (54), respectively, 
it holds that
\begin{equation}
 \bar{Y}^{-}\le \bar{Y}^{0,-}\le \bar{Y}^{0,+}\le \bar{Y}^{+} 
\end{equation}
and that
\begin{equation}
 \left\| \bar{Y}^{\pm}- \bar{Y}^{0,\pm}\right\|_{\beta,T}+
 \left\| \bar{Z}^{\pm}- \bar{Z}^{0,\pm}\right\|_{\beta,T}=O(\epsilon)
\end{equation}
as $\epsilon\to 0$
in both $+$ and $-$ cases.
\end{thm}

\begin{proof}[Sketch]
The relation (55) is easily seen from the comparison theorem of BSDEs.
To see (56), we can apply the continuity (and the differentiability) 
results with their proofs
with respect to parameterized BSDEs, shown in El Karoui et al.\ (2000)
(see Proposition~2.4 and its proof in \cite{EPQ} for the details).
\end{proof}
Combining Theorems~4 and 5, we see the following.
\begin{cor}
Assume Conditions~(19), (44), (48), and (49). 
Then $Y^{0,-}(0)$ and $Y^{0,+}(0)$ are arbitrage-free prices at time $0$
for the derivative security given in Definition~2. 
\end{cor}
The above corollary implies that 
$Y^{0,\pm}(0)$ may be regarded as 
approximated prices of the derivative security 
for the writer and her counterparty, which 
prohibit the existence of an arbitrage opportunity.
Because BSDEs for $(Y^{0,\pm},Z^{0,\pm})$ are linear,\footnote{That is, the drivers 
$f^{0,\pm}(t,y,z,u_1,u_2;\hat{v})$ are linear with respect to $(y,z,u_1,u_2)$.}
we obtain the closed-form expressions for $Y^{0,\pm}$ as follows.
Let us introduce the probability measure $\tilde{\mathbb P}_T$
on $(\Omega, {\mathcal F}_T)$ by
\[
d\tilde{\mathbb P}_T \bigm|_{{\mathcal F}_t}
={\mathcal E}(t) d{\mathbb P}\bigm|_{{\mathcal F}_t},
\quad t\in [0,T],
\]
where
\begin{align*}
 {\mathcal E}(t)
:=\exp\biggl[
\int_0^t \left\{ r_r^0(u)-r_D(u)\right\} 
{\bf 1}^\top (\sigma(u)^{-1})^{\top} dW(u) 
-\frac{1}{2}\int_0^t
\left\{ r_r^0(u)-r_D(u)\right\}^2 \left| \sigma(u)^{-1}{\bf 1}\right|^2 du
\biggr].
\end{align*}
We denote 
the expectation with respect to $\tilde{\mathbb P}_T$ 
conditioned by ${\mathcal F}_t$ 
by $\tilde{\mathbb E}_t\left[ (\cdots) \right]
=\tilde{\mathbb E}\left[ (\cdots)| {\mathcal F}_t\right]$.
Recall that
\[
 \tilde{W}(t):=W(t)
- \int_0^t \left\{ r_r^0(u)-r_D(u)\right\} \sigma(u)^{-1}{\bf 1}du,
\quad t\in [0,T]
\]
is a $(\tilde{\mathbb P}_T,{\mathbb F})$-Brownian motion
by the Maruyama--Girsanov theorem, and 
on $\left( \Omega, {\mathcal F},\tilde{\mathbb P}_T, {\mathbb F}\right)$
the risky asset price process $S$ has the dynamics
\[
 d{S}(t)
=\text{\rm diag}(S(t))
\left\{
\sigma(t)d\tilde{W}(t)+r_r^0(t) {\bf 1}dt
\right\},
\quad S(0)\in {\mathbb R}^n_{++}.
\]
Also, we denote 
\[
{\rm DF}_r(t,u):=\exp\left\{ -\int_t^u r(s)ds\right\}
\]
for the process $r:=(r(t))_{t\in [0,T]}$.
We then obtain the following.
\begin{prop}
The following representation holds:
\begin{equation}
 \bar{Y}^{0,\pm}(t)= {\rm V}(t) 
+ {\rm VA}_1(t)+{\rm VA}_2(t)
+{\rm VA}_3(t)
+{\rm VA}_4(t)
+{\rm VA}^{\pm}_5(t).
\end{equation}
Here, 
\begin{align*}
{\rm V}(t):=&\tilde{\mathbb E}_t\left[ {\rm DF}_{r_f^0}(t,T) \xi_T\right], \\
{\rm VA}_1(t):=& \tilde{\mathbb E}_t\left[ 
\int_t^T {\rm DF}_R(t,u)h_1(u) \hat{\phi}_1(u)du 
\right], \\
{\rm VA}_2(t):=& \tilde{\mathbb E}_t\left[ 
\int_t^T {\rm DF}_R(t,u)h_2(u) \hat{\phi}_2(u)du 
\right], \\
{\rm VA}_3(t):=& -\tilde{\mathbb E}_t\left[ 
\int_t^T {\rm DF}_R(t,u)
\left\{(r_f^0-r_D)\left( \hat{\phi}_1+\hat{\phi}_2 \right)\right\}(u)du
\right], \\
{\rm VA}_4(t):=& \tilde{\mathbb E}_t\left[ 
\int_t^T {\rm DF}_R(t,u)
\left\{ (r_r^0-r_D)
\left( \hat{\phi}_1\sigma_I+\hat{\phi}_2\sigma_C\right)\sigma^{-1}{\bf 1}
\right\}(u)du
\right], \\
{\rm VA}_5^{\pm}(t):=&\alpha \tilde{\mathbb E}_t\left[ 
\int_t^T {\rm DF}_R(t,u)
\left\{ 
\left( r_f^0-r_{col}^\pm\right)\hat{V}^+
- \left( r_f^0-r_{col}^\mp\right)\hat{V}^-
\right\}(u)du
\right],
\end{align*}
where we define
\begin{align*}
\hat{\phi}_i:=&\phi_i -{V} \quad\text{for $i=1,2$,\quad and} \\ 
R:=&r_D-\left(r_f^0-r_D\right) 
+\left\{(r_r^0-r_D)\left(
\sigma_I+\sigma_C\right)(\sigma)^{-1}{\bf 1}\right\}
+h_1+h_2.
\end{align*}
\end{prop}

\begin{proof}
Using the representation formula for linear BSDE 
(e.g., see Proposition~2.2 of \cite{EPQ}), we see that
\[
 \bar{Y}^{0,\pm}(t)= 
\bar{\rm V}(t) 
+ \overline{\rm VA}_1(t)
+\overline{\rm VA}_2(t)
+\overline{\rm VA}_3(t)
+\overline{\rm VA}_4(t)
+{\rm VA}^{\pm}_5(t),
\]
where 
\begin{align*}
\bar{\rm V}(t):=&\tilde{\mathbb E}_t\left[ {\rm DF}_{R}(t,T) \xi_T\right], \\
\overline{\rm VA}_1(t):=& \tilde{\mathbb E}_t\left[ 
\int_t^T {\rm DF}_R(t,u)h_1(u) {\phi}_1(u)du 
\right], \\
\overline{\rm VA}_2(t):=& \tilde{\mathbb E}_t\left[ 
\int_t^T {\rm DF}_R(t,u)h_2(u) {\phi}_2(u)du 
\right], \\
\overline{\rm VA}_3(t):=& -\tilde{\mathbb E}_t\left[ 
\int_t^T {\rm DF}_R(t,u)
\left\{(r_f^0-r_D)\left( {\phi}_1+{\phi}_2 \right)\right\}(u)du
\right], \\
\overline{\rm VA}_4(t):=& \tilde{\mathbb E}_t\left[ 
\int_t^T {\rm DF}_R(t,u)
\left\{ (r_r^0-r_D)
\left( {\phi}_1\sigma_I+{\phi}_2\sigma_C\right)\sigma^{-1}{\bf 1}
\right\}(u)du
\right].
\end{align*}
Furthermore, we see that
\begin{align*}
&\left[
{\rm VA}_1+{\rm VA}_2+{\rm VA}_3+{\rm VA}_4
- \overline{\rm VA}_1-\overline{\rm VA}_2-\overline{\rm VA}_3
-\overline{\rm VA}_4\right](t) \\
=&-\tilde{\mathbb E}_t
\left[ \int_t^T {\rm DF}_R(t,u) {\rm V}(u) 
\left\{ R(u)-r_f^0(u)\right\} du
\right] \\
=&- \tilde{\mathbb E}_t\left[\int_t^T {\rm DF}_R(t,u) 
\tilde{\mathbb E}_u \left[ {\rm DF}_{r_f^0}(u,T)\xi_T\right]
\left\{ R(u)-r_f^0(u)\right\} du\right] \\
=&\tilde{\mathbb E}_t\left[ 
{\rm DF}_{r_f^0}(t,T)\xi_T
\int_t^T 
\frac{\partial}{\partial u} {\rm DF}_{R-r_f^0}(t,u) du\right] \\
=&\tilde{\mathbb E}_t\left[ {\rm DF}_{r_f^0}(t,T)
\left\{ {\rm DF}_{R-r_f^0}(t,T)-1\right\}\xi_T\right] \\
=&\tilde{\mathbb E}_t\left[ 
\left\{ {\rm DF}_{R}(t,T)-{\rm DF}_{r_f^0}(t,T)\right\}\xi_T\right] 
=\bar{\rm V}(t)-{\rm V}(t), 
\end{align*}
hence the proof is complete.
\end{proof}
\begin{rem}
Suppose that $r_r^0\equiv r_f^0\equiv r_D$ holds. 
In this case, $\tilde{\mathbb P}_T\equiv {\mathbb P}$
and ${\rm V}\equiv \hat{V}$ follow.
Furthermore, consider 
$\phi_i(t):=\varphi_i\left(\hat{V}(t)\right)$, 
where (17) is employed for $i=1,2$. 
Then, in (57), 
${\rm VA}_3\equiv {\rm VA}_4\equiv 0$, and
$-{\rm VA}_1$, ${\rm VA}_2$, and ${\rm VA}_5^\pm$ are called 
the debt valuation adjustment (DVA), the credit valuation adjustment (CVA), 
and the collateral valuation adjustment (ColVA), respectively, 
which are popularly used XVA terms in practice for 
the valuation adjustment in the pricing of derivative securities.
Concretely, DVA, CVA, and ColVA at time $t$ are written as
\begin{align*}
{\rm DVA}(t):=&-
{\mathbb E}_t\left[ 
\int_t^T {\rm DF}_{r_D+h_1+h_2}(t,u)h_1(u) \hat{\phi}_1(u)du 
\right], \\
{\rm CVA}(t):=& {\mathbb E}_t\left[ 
\int_t^T {\rm DF}_{r_D+h_1+h_2}(t,u)h_2(u) \hat{\phi}_2(u)du 
\right], \\
{\rm ColVA}^\pm(t):=&
{\mathbb E}_t\Biggl[
\int_t^T {\rm DF}_{r_D+h_1+h_2}(t,u)
\left\{ 
\left( r_D-r_{col}^\pm\right)\alpha\hat{V}^+
- \left( r_D-r_{col}^\mp\right)\alpha\hat{V}^-
\right\}(u)du
\Biggr],
\end{align*}
respectively, 
where we denote ${\mathbb E}_t[(\cdots)]:={\mathbb E}\left[ (\cdots)
 |{\mathcal F}_t\right]$.
Further, 
\[
{\rm FVA}(t):=
{\mathbb E}_t \left[ \int_t^T 
{\rm DF}_{r_D + h_1+h_2}(t,u)
\left\{(r_f^0-r_D)\left( {\phi}_1+{\phi}_2 \right)\right\}(u)du
\right], 
\]
called the funding valuation adjustment (FVA) at time $t$, 
is another popularly used adjustment term in practice, 
which reflects the funding cost of uncollateralised derivatives 
above the riskfree rate of return.
We can roughly relate 
these XVA terms with the correction terms in Proposition 3 as follows: 
Let $r_r^0\equiv r_D$,\footnote{In practice,
the difference $r_r^0-r_D$ seems to have been usually ignored. 
} which implies ${\rm VA}_4\equiv 0$. 
Further, suppose $r_f^0 \approx r_D$. Then, 
we may interpret as 
\begin{align*}
{\rm DVA}\approx& -{\rm VA}_1, \\
{\rm CVA}\approx& {\rm VA}_2,  \\
{\rm ColVA}^\pm \approx &{\rm VA}_5^\pm, 
\end{align*}
and 
\[
{\rm FVA}\approx {\rm VA}_3, 
\]
or 
\[
{\rm FVA} \approx
{\rm VA}_3
+\left( {\rm VA}_1+{\rm DVA}\right)
+\left( {\rm VA}_2-{\rm CVA}\right)
+\left( {\rm ColVA}^\pm-{\rm VA}_5^\pm\right).
\]
For other theoretical studies on the valuation adjustments
and related interpretation of XVA used in practice, we refer to 
Brigo et al.\ (2020) and the reference therein. 
Also, for comprehensive information on 
XVA issue and expanding related issues (e.g., computational issue), 
see for example 
Gregory (2015) and Glau et al.\ (2016), 
and the references therein, which are still nonexhaustive.
\end{rem}

\subsection{Perturbed BSDEs}

As we see in Theorem~5 and Corollary~1, under certain conditions, 
$Y^{0,+}(t) (<Y^+(t))$, which is a zeroth-order approximation 
of the minimal hedging cost $Y^{+}(t)$, 
is an arbitrage-free price for the writer at time $t$. 
In this subsection, we try to improve our hedging strategy 
by using a first-order approximation. 
Using the solution to BSDE (53), consider the linear BSDE 
\begin{equation}
\begin{split}
-dY^{1,\pm}(t)
=&
f^{0}\left( t, Y^{1,\pm}(t), Z^{1,\pm}(t), 
U^{1,\pm}_1(t), U^{1,\pm}_2(t)\right) dt \\
+&f^{1,\pm}\left( t, Y^{0,\pm}(t), Z^{0,\pm}(t), 
U^{0,\pm}_1(t), U^{0,\pm}_2(t), \hat{V}(t) \right)dt  \\
-&Z^{1,\pm}(t)dW(t)
-U^{1,\pm}_1(t)dM_1(t)-U^{1,\pm}_2(t) dM_2(t), \\
Y^{1,\pm}(\tau_1\wedge\tau_2\wedge T)=&0
\end{split}
\end{equation}
on $(\Omega,{\mathcal F},{\mathbb P},{\mathbb G})$, 
where
\begin{multline*}
f^{1,\pm}(t,y,z,u_1,u_2;\hat{v})
:=\pm \epsilon_f(t)\left| y+u_1+u_2-\alpha \hat{v}\right| \\
\pm \epsilon_r(t) 
\left|
\left\{  z^\top +u_1\sigma_I(t) +u_2\sigma_C(t)\right\}
\sigma(t)^{-1}{\bf 1}\right|. 
\end{multline*}
Furthermore, using the solution to BSDE (54), consider the linear BSDE 
\begin{equation}
\begin{split}
-d\bar{Y}^{1,\pm}(t)
=&
\bar{f}^{0}\left( t, \bar{Y}^{1,\pm}(t), \bar{Z}^{1,\pm}(t); 
\phi_1(t),\phi_2(t) \right)dt  \\
+&\bar{f}^{1,\pm}\left( t, \bar{Y}^{0,\pm}(t), \bar{Z}^{0,\pm}(t); 
\hat{V}(t), \phi_1(t),\phi_2(t) \right)dt  \\
-&\bar{Z}^{1,\pm}(t)dW(t), \\
\bar{Y}^{1,\pm}(T)=&0
\end{split}
\end{equation}
on $(\Omega,{\mathcal F},{\mathbb P},{\mathbb F})$, 
where
\begin{align*}
\bar{f}^{0}(t,y,z;p_1,p_2):=&
f^0\left( t,y,z,p_1-y,p_2-y\right), \\
\bar{f}^{1,\pm}(t,y,z;\hat{v},p_1,p_2)
:=&
\pm \epsilon_f(t) \left| y+ (p_1-y)+ (p_2-y)-\alpha \hat{v}\right| \\
&\pm \epsilon_r(t) \left|
\left\{  z^\top +(p_1-y)\sigma_I(t) +(p_2-y)\sigma_C(t)\right\}
\sigma(t)^{-1}{\bf 1}\right|. 
\end{align*}
Using a similar technique to that used in the proof of Theorem~5, 
we can show the following.
\begin{prop}
It holds that for any sufficiently large $\beta>0$, 
\begin{align*}
\| \bar{Y}^\pm-\left( \bar{Y}^{0,\pm}+\bar{Y}^{1,\pm} \right)\|_{\beta,T}
+\| \bar{Z}^\pm-\left( \bar{Z}^{0,\pm}+\bar{Z}^{1,\pm} \right)\|_{\beta,T}
=&O(\epsilon^2) 
\end{align*}
as $\epsilon\to 0$, where we assume (52).
\end{prop}

\section*{Acknowledgements}
The authors are grateful to an anonymous referee 
for valuable comments and helpful suggestions.
%
%

\end{document}